\documentclass[envcountsame,runningheads]{llncs}

\usepackage{ifthen}

\newif\ifTR
\TRtrue

\ifTR
\else
\fi

\ifTR
\else
  \usepackage[paperheight=235mm, paperwidth=155mm,textwidth=12.2cm,textheight=19.3cm]{geometry}
\fi

\usepackage{colortbl}
\usepackage{xspace}
\usepackage{amsmath}
\usepackage{amssymb}
\usepackage{xspace}
\usepackage[dvipsnames]{xcolor}
\usepackage{subcaption}
\usepackage{dcolumn}
\usepackage{tikz}
\usepackage{booktabs}
\usepackage{graphicx}
\usepackage[linesnumbered,ruled,noend,vlined]{algorithm2e}
\usepackage{paralist}
\usepackage{bbding}

\makeatletter
\RequirePackage[bookmarks,unicode,colorlinks=true]{hyperref}%
   \def\@citecolor{blue}%
   \def\@urlcolor{blue}%
   \def\@linkcolor{blue}%

\def\orcidID#1{\smash{\href{http://orcid.org/#1}{\protect\raisebox{-1.25pt}{\protect\includegraphics{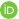}}}}}
\makeatother

\usepackage[capitalise,noabbrev]{cleveref}
\usepackage{environ}
\usepackage{wrapfig}
\usepackage{thm-restate}
\usepackage{xcolor}
\usepackage{tcolorbox}

\usepackage{mathtools}
\usepackage{newtxmath}
\usepackage{newtxtext}

\usepackage{fontawesome}

\NewEnviron{killcontents}{}

\usetikzlibrary{automata,positioning,trees,shapes,petri,arrows,snakes,backgrounds,calc,fit,arrows.meta,shadows,math,decorations.markings}
\usetikzlibrary{decorations.pathmorphing}
\usetikzlibrary{shapes.multipart,arrows,automata}
\usetikzlibrary{shapes.geometric}
\usetikzlibrary{positioning}
\usetikzlibrary{calc,patterns,patterns.meta}

\usepackage{lineno}
\newcounter{claimcounter}
\renewenvironment{claim}{\refstepcounter{claimcounter}{\medskip\noindent \underline{Claim \theclaimcounter:}}\itshape}{\smallskip}
\crefname{claimcounter}{Claim}{Claims}

\Crefname{algocf}{Algorithm}{Algorithms}
\crefname{figure}{Fig.}{Figs.}

\InputIfFileExists{cutmargins.tex}{%
	\pdfpagesattr{/CropBox [92 112 523 758]} 
}{%
}

\newcommand{\vh}[1]{\textcolor{orange}{\ifmmode \text{[#1]}\else [VH: #1] \fi}}
\newcommand{\ol}[1]{\textcolor{blue}{\ifmmode \text{[OL: #1]}\else [OL: #1] \fi}}
\newcommand{\bs}[1]{\textcolor{ForestGreen}{\ifmmode \text{[BS: #1]}\else [BS: #1] \fi}}

\renewcommand{\vh}[1]{}
\renewcommand{\ol}[1]{}
\renewcommand{\bs}[1]{}

\newcommand{\blinded}[1]{\ifx\blindreview\undefined #1 \else \textcolor{black!65}{[blinded for review]}\fi}

\newcommand{\sofi}[0]{(S,O,f,i)}

\newcommand{\sof}[0]{(S,O,f)}

\newcommand{\rank}[0]{\mathit{rank}}
\newcommand{\reach}[0]{\mathit{reach}}
\newcommand{\rankof}[1]{\mathit{rank}(#1)}
\newcommand{\level}[0]{\mathit{level}}

\newcommand{\levelw}[0]{\mathit{level}_\word}
\newcommand{\levelwof}[1]{\mathit{level}_\word(#1)}

\newcommand{\rankwof}[1]{\rank_{\word}(#1)}

\newcommand{\infofs}[1]{\inf_Q(#1)}
\newcommand{\infoft}[1]{\inf_\delta(#1)}

\newcommand{\variant}[2]{#1 \lhd #2}

\newcommand{\restrof}[2]{#1 \raisebox{-.5ex}{$|$}_{#2}}

\newcommand{\aut}[0]{\mathcal{A}}
\newcommand{\but}[0]{\mathcal{B}}

\newcommand{\M}[0]{\mathcal{M}}

\newcommand{\E}[0]{\mathcal{E}}
\newcommand{\cT}{\mathcal{T}}
\newcommand{\cR}{\mathcal{R}}

\newcommand{\cP}{\mathcal{P}}

\newcommand{\transover}[1]{\overset{#1}{\rightarrow}}
\newcommand{\ltr}[1]{\transover{#1}}

\newcommand{\word}[0]{\alpha}
\newcommand{\wordof}[1]{\seqof{\word}{#1}}
\newcommand{\dagg}[0]{\mathcal{G}}
\newcommand{\dagof}[1]{\dagg_{#1}}
\newcommand{\dagw}[0]{\dagof{\word}}

\newcommand{\seqof}[2]{#1_{#2}}
\newcommand{\dagwiof}[1]{\dagw^{#1}}

\newcommand{\algschewe}[0]{\textsc{Schewe}\xspace}

\newcommand{\algdeelev}[0]{\textsc{DeElev}\xspace}

\newcommand{\val}[0]{\mathbb{V}}

\newcommand{\lang}[0]{\mathcal{L}}
\newcommand{\langof}[1]{\lang(#1)}
\newcommand{\langautof}[2]{\lang_{#1}(#2)}

\newcommand{\numsetof}[1]{[#1]}

\newcommand{\accstates}{Q_F}
\newcommand{\acctrans}{\delta_F}
\newcommand{\trans}{\delta}

\newcommand{\R}{\mathcal{R}}

\newcommand{\claimqed}[0]{\hfill $\blacksquare$}
\newenvironment{claimproof}[1]{\par\noindent\underline{Proof:}\space#1}{\claimqed}

\newcommand{\pspace}[0]{\textsc{PSpace}\xspace}
\newcommand{\bigO}[0]{\mathcal{O}}
\newcommand{\bigOof}[1]{\bigO(#1)}

\newcommand{\ranker}[0]{\textsc{Ranker}\xspace}
\newcommand{\rankermaxrank}[0]{\ranker_{\textsc{MaxR}}\xspace}

\newcommand{\rankerold}[0]{\textsc{Ranker}_{\textsc{Old}}\xspace}
\newcommand{\rabit}[0]{\textsc{Rabit}\xspace}
\newcommand{\spot}[0]{\textsc{Spot}\xspace}
\newcommand{\seminator}[0]{\textsc{Seminator}~2\xspace}
\newcommand{\goal}[0]{\textsc{Goal}\xspace}
\newcommand{\roll}[0]{\textsc{Roll}\xspace}
\newcommand{\fribourg}[0]{\textsc{Fribourg}\xspace}
\newcommand{\piterman}[0]{\textsc{Piterman}\xspace}
\newcommand{\safra}[0]{\textsc{Safra}\xspace}
\newcommand{\autfilt}[0]{\texttt{autfilt}\xspace}
\newcommand{\ltldstar}[0]{\textsc{LTL2dstar}\xspace}
\newcommand{\goalmark}[0]{\scriptsize \faSoccerBallO}

\definecolor{rowgray}{gray}{0.85}

\newcommand{\condensof}[1]{\mathit{cond}(#1)}

\newcommand{\typemax}[0]{\mathrm{typemax}}
\newcommand{\expr}[0]{e}
\newcommand{\exprof}[1]{\expr_{#1}}

\newcommand{\trub}[0]{TRUB\xspace}
\newcommand{\skeletof}[1]{\mathcal{K}_{#1}}
\newcommand{\upd}[0]{\mathit{up}}
\newcommand{\tight}[0]{\textsc{Tight}\xspace}
\newcommand{\waiting}[0]{\textsc{Waiting}\xspace}

\newcommand{\infunc}[0]{\mathit{inner}}
\newcommand{\lG}[0]{\mathsf{G}}
\newcommand{\lF}[0]{\mathsf{F}}

\newcommand{\decf}[0]{\mathit{dec}}
\newcommand{\decof}[1]{\decf(#1)}

\newcommand{\autex}[0]{\aut_{\mathit{ex}}}

\newcommand{\scclabel}[2]{\tikz[baseline,anchor=base]{\node[rectangle,draw,rounded corners] {#1{:}#2};}}

\newcommand{\detcircle}{node[left,shape=circle,inner sep=1pt,scale=0.7,fill=orange!80!black,text=white,draw=orange!80!black,xshift=-1mm,yshift=1mm]  {$2$}}
\newcommand{\textdetcircle}{%
\tikz[baseline,anchor=base,scale=0.5]{ \draw \detcircle;}}

\newcommand{\iwcircle}{node[right,shape=circle,inner sep=1pt,scale=0.7,fill=blue!80!black,text=white,draw=blue!80!black,xshift=1mm,yshift=1mm]  {$2$}}
\newcommand{\textiwcircle}{\tikz[baseline,anchor=base,scale=0.5]{ \draw \iwcircle;}}

\newcommand{\evenceil}[1]{\lfloor\!\!\lfloor #1 \rfloor\!\!\rfloor}

\newcommand{\maxr}[0]{\textit{max-succ-rank}}
\newcommand{\maxrsa}[0]{\maxr_S^a}
\newcommand{\maxrsaof}[2]{\maxr^a_{#1}(#2)}

\newcommand{\beginexample}[0]{%
\medskip\noindent\stepcounter{theorem}\textit{Example~\thetheorem}. }

\makeatletter
\newcommand{\monus}{\mathbin{\text{\@dotminus}}}
\newcommand{\@dotminus}{%
  \ooalign{\hidewidth\raise1ex\hbox{.}\hidewidth\cr$\m@th-$\cr}%
}
\makeatother

\newcommand{\transconsist}[0]{%
  \begin{tikzpicture}[line width=0.6pt,transform shape,scale=0.7]
    \draw[->] (0,0) -- (1.5em,0);
    \draw (0.6em,0ex) circle (0.7ex);
  \end{tikzpicture}%
}

\newcommand{\stirl}[2]{{#1\brace #2}}

\newcommand{\deelevof}[1]{#1^{\bullet}}

\newcommand{\dsrandom}[0]{\underline{\texttt{random}}\xspace}
\newcommand{\dsltl}[0]{\underline{\texttt{LTL}}\xspace}
\newcommand{\dsall}[0]{\underline{\texttt{all}}\xspace}

\title{Sky Is Not the Limit}
\subtitle{Tighter Rank Bounds for Elevator Automata in\\ B\"{u}chi Automata Complementation\ifTR{\\ (Technical Report)}\fi}

\titlerunning{Sky Is Not the Limit: Tighter Rank Bounds in B\"{u}chi Automata Complementation}

\author{
  Vojt\v{e}ch Havlena\orcidID{0000-0003-4375-7954} \and
  Ond\v{r}ej Leng\'{a}l\orcidID{0000-0002-3038-5875}\Envelope \and
	Barbora \v{S}mahl\'{i}kov\'{a}\orcidID{0000-0002-1184-4669}
  }

\authorrunning{
  Vojt\v{e}ch Havlena \and
  Ond\v{r}ej Leng\'{a}l \and
	Barbora \v{S}mahl\'{i}kov\'{a}
}

\institute{
  Faculty of Information Technology,
  Brno University of Technology,
  Brno,
  Czech Republic
}

\begin{document}

\maketitle

\begin{abstract}
We propose several heuristics for mitigating one of the main causes of
combinatorial explosion in rank-based complementation of
B\"{u}chi automata (BAs):
unnecessarily high bounds on the ranks of states.
First, we identify \emph{elevator automata}, which is a~large class of BAs
(generalizing semi-deterministic BAs), occurring often in practice, where ranks
of states are bounded according to the structure of strongly connected
components.
The bounds for elevator automata also carry over to general BAs that contain
elevator automata as a~sub-structure.
Second, we introduce two techniques for refining bounds on the ranks of BA
  states using data-flow analysis of the automaton.
We implement out techniques as an extension of the tool \ranker for BA
complementation and show that they indeed greatly prune the generated
state space, obtaining significantly better results and outperforming other
state-of-the-art tools on a~large set of benchmarks.
\end{abstract}

\vspace{-8.0mm}
\section{Introduction}
\vspace{-0.0mm}

\emph{B\"{u}chi automata} (BA) complementation has been a~fundamental problem
underlying many applications since it was introduced in
1962 by B\"{u}chi~\cite{buchi1962decision,HLS-S1S} as an essential part of a~decision procedure for
a~fragment of the second-order arithmetic.
BA~complementation has been used as a~crucial part of, e.g.,
termination analysis of
programs~\cite{fogarty2009buchi,heizmann2014termination,ChenHLLTTZ18} or
decision procedures for various logics, such as S1S~\cite{buchi1962decision},
the first-order logic of Sturmian words~\cite{pecan}, or the temporal logics
ETL and QPTL~\cite{sistla1987complementation}.
Moreover, BA complementation also underlies BA inclusion and equivalence
testing, which are essential instruments in the BA toolbox.
Optimal algorithms, whose output asymptotically matches the lower bound of
$(0.76n)^n$~\cite{yan} (potentially modulo a~polynomial factor), have been
developed~\cite{Schewe09,fribourg}.
For a~successful real-world use, asymptotic optimality is, however, not enough and these
algorithms need to be equipped with a~range of optimizations to make them
behave better than the worst case on BAs occurring in practice.



\enlargethispage{2mm}

In this paper, we focus on the so-called \emph{rank-based} approach to
complementation, introduced by Kupferman and Vardi~\cite{KupfermanV01},
further improved with the help of Friedgut~\cite{FriedgutKV06}, and finally
made optimal by Schewe~\cite{Schewe09}.
The construction stores in a~macrostate partial information about all
runs of a~BA~$\aut$ over some word~$\word$.
In addition to tracking states that~$\aut$ can be in (which is sufficient, e.g., in
the determinization of NFAs), a~macrostate also stores a~guess of the rank of each
of the tracked states in the \emph{run DAG} that captures all these runs.
The guessed ranks impose restrictions on how the future of a~state might
look like (i.e., when~$\aut$ may accept).
The number of macrostates in the complement depends combinatorially on the
maximum rank that occurs in the macrostates.
The constructions in~\cite{KupfermanV01,FriedgutKV06,Schewe09} provides only coarse
\mbox{bounds on the maximum ranks}.

A~way of decreasing the maximum rank has been suggested
in~\cite{GurumurthyKSV03} using a~\pspace (and, therefore, not really
practically applicable) algorithm (the problem of finding the optimal rank is \pspace-complete).
In our previous paper~\cite{HavlenaL2021}, we have identified several basic
optimizations of the construction that can be used to refine the
\emph{tight-rank upper bound} (TRUB) on the maximum ranks of states.
In this paper, we push the applicability of rank-based techniques much further
by introducing two novel lightweight techniques for refining the TRUB,
thus significantly reducing the generated state space.

Firstly, we introduce a~new class of the so-called \emph{elevator automata},
which occur quite often in practice (e.g., as outputs of natural algorithms for
translating LTL to BAs).
Intuitively, an elevator automaton is a~BA whose strongly connected components
(SCCs) are all either inherently weak\footnote{%
An SCC is inherently weak if it either contains no accepting states or, on the
other hand, all cycles of the SCC contain an accepting state.
}
or deterministic.
Clearly,
the class substantially generalizes the popular inherently
weak~\cite{BoigelotJW01} and semi-deterministic
BAs~\cite{CourcoubetisY88,semiDetComplementation,seminator}).
The structure of elevator automata allows us to provide tighter estimates of
the TRUBs, not only for elevator automata \emph{per se}, but also for BAs where
elevator automata occur as a~sub-structure (which is even more common).
Secondly, we propose a~lightweight technique, inspired by data flow analysis,
allowing to propagate rank restriction along the skeleton of the complemented
automaton, obtaining even tighter TRUBs.
We also extended the optimal rank-based algorithm to transition-based
BAs (TBAs).

We implemented our optimizations within the \ranker tool~\cite{ranker} and
evaluated our
approach on thousands of hard automata from the literature (15\,\% of
them were elevator automata that were not semi-deterministic, and many more
contained an elevator sub-structure). Our techniques drastically reduce the
generated state space; in many cases we even achieved exponential improvement
compared to the optimal procedure of Schewe and our previous heuristics.
The new version of \ranker gives a smaller
complement in the majority of cases of hard automata than other
state-of-the-art tools.

\vspace{-3.0mm}
\section{Preliminaries}
\vspace{-2.0mm}

\paragraph{Words, functions.}
We fix a~finite nonempty alphabet~$\Sigma$ and the first infinite ordinal
$\omega = \{0, 1, \ldots\}$. For $n\in\omega$, by $\numsetof{n}$ we denote the
set $\{ 0, \dots, n \}$. For $i \in \omega$ we use $\evenceil{i}$ to denote
the largest even number smaller of equal to~$i$, e.g., $\evenceil{42} =
\evenceil{43} = 42$.
An (infinite) word~$\word$ is
represented as a~function $\word\colon \omega \to \Sigma$ where the $i$-th
symbol is denoted as $\wordof i$. We~abuse notation and sometimes also
represent~$\word$ as an~infinite sequence $\word = \wordof 0 \wordof 1 \dots$
We~use~$\Sigma^\omega$ to denote the set of all infinite words over~$\Sigma$.
For a~(partial) function $f\colon X\to Y$ and a set $S\subseteq X$, we define $f(S) = \{ f(x)
\mid x\in S \}$. Moreover, for $x\in X$ and $y\in Y$, we use $f \triangleleft \{
x\mapsto y \}$ to denote the function $(f \setminus \{x \mapsto f(x)\}) \cup \{
x\mapsto y \}$.

\vspace{-1.0mm}
\paragraph{B\"{u}chi automata.}
A~(nondeterministic transition/state-based) \emph{B\"{u}chi automaton} (BA)
over~$\Sigma$ is a~quadruple $\aut = (Q, \trans, I, \accstates \cup \acctrans)$
where $Q$ is a~finite set of \emph{states}, $\trans\colon Q \times \Sigma \to
2^Q$ is a~\emph{transition function}, $I \subseteq Q$ is the sets of
\emph{initial} states, and $\accstates \subseteq Q$ and $\acctrans \subseteq
\trans$ are the sets of \emph{accepting states} and \emph{accepting
transitions} respectively.
We sometimes treat~$\trans$ as a~set of transitions $p \ltr a q$, for instance,
we use $p \ltr a q \in \trans$ to denote that $q \in \trans(p, a)$. Moreover, we
extend $\trans$ to sets of states $P \subseteq Q$ as $\trans(P, a) = \bigcup_{p
\in P} \trans(p,a)$, and to sets of symbols $\Gamma \subseteq\Sigma$ as
$\trans(P,\Gamma) = \bigcup_{a\in\Gamma}\trans(P,a)$. We define the inverse
transition function as $\trans^{-1} = \{ p \ltr a q \mid q \ltr a p \in \trans\}$.
The notation $\restrof \trans S$ for $S \subseteq Q$ is used to denote the
restriction of the transition function $\trans \cap (S \times \Sigma \times S)$.
Moreover, for $q \in Q$, we use $\aut[q]$ to denote the BA $(Q, \trans, \{q\},
\accstates \cup \acctrans)$.

A~\emph{run}
of~$\aut$ from~$q \in Q$ on an input word $\word$ is an infinite sequence $\rho\colon
\omega \to Q$ that starts in~$q$ and respects~$\trans$, i.e., $\rho_0 = q$ and
$\forall i \geq 0\colon \rho_i \ltr{\wordof i}\rho_{i+1} \in \trans$.
Let $\infofs \rho$ denote the states occurring in~$\rho$ infinitely often and
$\infoft \rho$ denote the transitions occurring in~$\rho$ infinitely often.
The run~$\rho$ is called \emph{accepting} iff $\infofs \rho \cap \accstates \neq \emptyset$ or
$\infoft \rho \cap \acctrans \neq \emptyset$.

A~word~$\word$ is accepted by~$\aut$ from a~state~$q \in Q$ if there is an
accepting run~$\rho$ of~$\aut$ from~$q$, i.e., $\rho_0 = q$. The set
$\langautof{\aut} q = \{\word \in \Sigma^\omega \mid \aut \text{ accepts } \word
\text{ from } q\}$ is called the \emph{language} of~$q$ (in~$\aut$). Given a~set
of states~$R \subseteq Q$, we define the language of~$R$ as $\langautof \aut R =
\bigcup_{q \in R} \langautof \aut q$ and the language of~$\aut$ as~$\langof \aut =
\langautof \aut I$.
We say that a~state $q \in Q$ is \emph{useless} iff $\langautof{\aut} q = \emptyset$.
%
If $\acctrans = \emptyset$, we call~$\aut$ \emph{state-based} and
if $\accstates = \emptyset$, we call~$\aut$ \emph{transition-based}.
In this paper, we fix a~BA $\aut = (Q, \trans, I, \accstates \cup \acctrans)$.

%

\vspace{-2mm}
\section{Complementing B\"{u}chi automata}\label{sec:rank-based-compl}
\vspace{-1mm}

In this section, we describe a~generalization of the rank-based complementation of state-based BAs
presented by Schewe in~\cite{Schewe09} to our notion of transition/state-based~BAs.
\ifTR
Proofs can be found in the Appendix.
\fi

\vspace{-2.0mm}
\subsection{Run DAGs}
\vspace{-1.0mm}

First, we recall the terminology from~\cite{Schewe09} (which is
a~minor modification of the one in~\cite{KupfermanV01}), which we use
in the paper. Let the \emph{run DAG} of~$\aut$ over
a~word~$\word$ be a~DAG (directed acyclic graph) $\dagw = (V,E)$ containing
vertices~$V$ and edges~$E$ such that
\vspace{-0mm}
\begin{itemize}
  \setlength{\itemsep}{0mm}
  \item  $V \subseteq Q \times \omega$ s.t. $(q, i) \in V$ iff there is
    a~run~$\rho$ of $\aut$ from $I$ over $\word$ with $\rho_i = q$,
  \item  $E \subseteq V \times V$ s.t.~$((q, i), (q',i')) \in E$ iff $i' = i+1$
    and $q' \in \trans(q, \wordof i)$.
\end{itemize}
\vspace{-0mm}

\noindent
Given $\dagw$ as above, we will write $(p, i) \in \dagw$ to denote that $(p, i)
\in V$.
A~vertex $(p,i) \in V$ is called
\emph{accepting} if $p$ is an accepting state and an
edge $((q, i), (q',i')) \in E$ is called \emph{accepting} if $q \ltr{\word_i}
q'$ is an accepting transition. A~vertex~$v \in \dagw$ is \emph{finite} if the
set of vertices reachable from $v$ is finite, \emph{infinite} if it is not
finite, and \emph{endangered} if it cannot reach an accepting vertex or an
accepting edge.

We assign ranks to vertices of run DAGs as follows:
Let $\dagwiof 0 = \dagw$ and~$j = 0$.
Repeat the following steps until the~fixpoint or for at most $2n + 1$ steps,
where $n = |Q|$.
\vspace{-0mm}
\begin{itemize}
  \setlength{\itemsep}{0mm}
  \item  Set $\rankwof{v} \gets j$ for all finite vertices $v$ of~$\dagwiof j$
    and let $\dagwiof{j+1}$ be $\dagwiof{j}$ minus the vertices with the
    rank~$j$.
  \item  Set $\rankwof{v} \gets j+1$ for all endangered vertices $v$
    of~$\dagwiof {j+1}$ and let $\dagwiof{j+2}$ be $\dagwiof{j+1}$ minus the
    vertices with the rank~$j+1$.
  \item  Set $j \gets j + 2$.
\end{itemize}
\vspace{-0mm}

\noindent
For all vertices~$v$ that have not been assigned a~rank yet, we assign
$\rankwof{v} \gets \omega$.


We define the \emph{rank of $\word$}, denoted as $\rankof \word$, as $\max\{\rankwof v \mid v \in \dagw\}$ and
the \emph{rank of $\aut$}, denoted as $\rankof \aut$, as $\max\{\rankof w \mid
w \in \Sigma^\omega \setminus \langof \aut \}$.

\begin{restatable}{lemma}{theMaxRank}\label{lemma:maxRank}
	If $\alpha \notin \langof{\aut}$, then $\rankof \word \leq 2|Q|$.
\end{restatable}



\vspace{-5.0mm}
\subsection{Rank-Based Complementation}
\vspace{-0.0mm}

In this section, we describe a~construction for complementing BAs developed in
the work of Kupferman and Vardi~\cite{KupfermanV01}---later improved by
Friedgut, Kupferman, and Vardi~\cite{FriedgutKV06}, and by
Schewe~\cite{Schewe09}---extended to our definition of BAs with accepting
states and transitions (see~\cite{HavlenaL2021} for a~step-by-step introduction).
The construction is based on the notion of tight level rankings storing
information about levels in run DAGs.
For a~BA~$\aut$ and $n = |Q|$, a~\emph{(level) ranking} is a~function $f\colon Q
\to \numsetof{2n}$ such that $f(\accstates) \subseteq \{0, 2,
\ldots, 2n\}$, i.e., $f$~assigns even ranks to accepting states of~$\aut$. For
two rankings~$f$ and~$f'$ we define $f \transconsist_S^a f'$ iff for each $q
\in S$ and $q'\in\trans(q,a)$ we have $f'(q') \leq f(q)$ and for each
$q''\in \acctrans(q,a)$ it holds $f'(q'') \leq \evenceil{f(q)}$.
The set of all rankings is denoted by~$\cR$.
For a~ranking~$f$, the \emph{rank} of~$f$ is defined as~$\rankof f = \max\{f(q)
\mid q \in Q\}$.
We use $f \leq f'$ iff for every state $q \in Q$ we have $f(q)
\leq f'(q)$ and we use $f < f'$ iff $f \leq f'$ and there is a~state $q \in Q$ with
$f(q) < f'(q)$.
For a~set of states $S \subseteq Q$, we call~$f$ to be $S$-\emph{tight} if
\begin{inparaenum}[(i)]
  \item  it has an odd rank~$r$,
  \item  $f(S) \supseteq \{1, 3, \ldots, r\}$, and
  \item  $f(Q\setminus S) = \{0\}$.
\end{inparaenum}
A~ranking is \emph{tight} if it is $Q$-tight; we use~$\cT$ to denote the set of
all tight rankings.

%

%

The original rank-based construction~\cite{KupfermanV01} uses macrostates of the
form $\sof$ to track all runs of~$\aut$ over~$\word$.
The $f$-component contains guesses of the ranks of states in~$S$ (which is
obtained by the classical subset construction) in the run DAG and the $O$-set is
used to check whether all runs contain only a~finite number of accepting states.
Friedgut, Kupferman, and Vardi~\cite{FriedgutKV06}
improved the construction by having~$f$ consider only tight rankings.
Schewe's construction~\cite{Schewe09} extends the macrostates to $\sofi$ with $i \in
\omega$ representing a~particular even rank such that~$O$ tracks states with
rank~$i$.
At~the cut-point
(a~macrostate with $O=\emptyset$) the value of $i$ is changed to $i+2$ modulo the
rank of~$f$. Macrostates in an accepting run hence iterate over all possible
values of~$i$. Formally, the complement of~$\aut = (Q, \trans, I, \accstates
\cup \acctrans)$ is given as the (state-based) BA $\algschewe(\aut) = (Q',
\trans', I', \accstates' \cup \emptyset)$, whose components are defined as
follows:
\begin{itemize}
  \item  $Q' = Q_1 \cup Q_2$ where
    \begin{itemize}
      \item  $Q_1 = 2^Q$ and
      \item  $Q_2 = \hspace*{-1mm}
        \begin{array}[t]{ll}
          \{(S,O,f, i) \in \hspace*{0mm}& 2^Q \times 2^Q \times \cT \times \{0, 2, \ldots, 2n - 2\} \mid  f \text{ is $S$-tight},
          \\ & O \subseteq S \cap f^{-1}(i)\},
        \end{array}$
    \end{itemize}
  \item  $I' = \{I\}$,
  \item  $\trans' = \trans_1 \cup \trans_2 \cup \trans_3$ where
    \begin{itemize}
      \item  $\trans_1\colon Q_1 \times \Sigma \to 2^{Q_1}$ such that $\trans_1(S, a) =
        \{\trans(S,a)\}$,
      \item $\trans_2 \colon Q_1 \times \Sigma \to 2^{Q_2}$ such that $\trans_2(S, a) =
        \{(S', \emptyset, f, 0) \mid S' = \trans(S, a),\linebreak f \text{ is } S'\text{-tight}\}$, and
      \item $\trans_3\colon Q_2 \times \Sigma \to 2^{Q_2}$ such that $(S', O', f', i') \in
        \trans_3((S, O, f, i), a)$ iff
          \begin{itemize}
            \item  $S' = \trans(S, a)$,
            \item  $f \transconsist_S^a f'$,
            \item  $\rankof f = \rankof{f'}$,
            \item  and
              \begin{itemize}[$\circ$]
                \item if $O = \emptyset$ then $i' = (i+2) \mod (\rankof{f'} +
                  1)$ and $O' = f'^{-1}(i')$, and
                \item  if $O \neq \emptyset$ then $i' = i$ and $O' = \trans(O,
                  a) \cap f'^{-1}(i)$; and
              \end{itemize}
          \end{itemize}
    \end{itemize}
    \vspace{1mm}
  \item  $\accstates' = \{\emptyset\} \cup ((2^Q \times \{\emptyset\} \times \cT \times
    \omega) \cap Q_2)$.
\end{itemize}

\noindent
We call the part of the automaton with states from~$Q_1$ the \emph{waiting}
part (denoted as \waiting), and the part corresponding to~$Q_2$ the \emph{tight} part (denoted as \tight).

\begin{restatable}{theorem}{theComplCorr}\label{the:schewe-corr}
	Let $\aut$ be a~BA.
  Then $\langof{\algschewe(\aut)} = \Sigma^\omega \setminus \langof{\aut}$.
\end{restatable}


\newcommand{\figschewe}[0]{
\begin{wrapfigure}[17]{r}{6.3cm}
\vspace*{-14mm}
\hspace*{-3mm}
\scalebox{0.96}{
  \begin{minipage}{6.5cm}
  	\centering
  	\begin{subfigure}[b]{1.0\textwidth}
  		\centering
  		\scalebox{0.9}{
  			\begin{tikzpicture}[->,>=stealth',shorten >=0pt,auto,node distance=1.5cm,
                    scale=0.8,transform shape,initial text={}]
  \tikzstyle{every state}=[inner sep=3pt,minimum size=5pt]
  \tikzstyle{empty}=[]
  \tikzstyle{initstate}=[fill=yellow!30]

  \node[state,initial,initstate] (r) {$r$};
  \node[state,accepting,right of=r] (s) {$s$};
  \node[state,right of=s] (t) {$t$};

  \path (r) edge[loop above]  node {$b$} (r)
        (r) edge  node {$b$} (s)
        (r) edge[bend right=60] node[below] {$b$} (t)
        (s) edge[loop above]  node {$a$} (s)
        (s) edge[bend left]  node {$a$} (t)
        (t) edge[bend left]  node {$a$} (s);

\end{tikzpicture}
  		}
  		\caption{BA $\aut$ over $\{ a,b \}$}
  		\label{fig:schewe-aut}
  	\end{subfigure}
  	~
  	\begin{subfigure}[b]{1.0\textwidth}
  		\centering
  		\scalebox{0.9}{
  			\begin{tikzpicture}[->,>=stealth',shorten >=0pt,auto,node distance=1.8cm,
                    scale=0.6,transform shape,initial text={}]
  \tikzstyle{every state}=[inner sep=3pt,minimum size=5pt,rectangle,rounded corners=1mm]
  \tikzstyle{empty}=[]
  \tikzstyle{initstate}=[fill=yellow!30]
  \tikzstyle{wobbly}=[decorate, decoration={snake,amplitude=.2mm,segment length=2mm,post length=1mm}]

  \tikzdeclarepattern{
    name=hatch,
    parameters={\hatchsize,\hatchangle,\hatchlinewidth},
    bounding box={(-.1pt,-.1pt) and (\hatchsize+.1pt,\hatchsize+.1pt)},
    tile size={(\hatchsize,\hatchsize)},
    tile transformation={rotate=\hatchangle},
    defaults={
      hatch size/.store in=\hatchsize,hatch size=5pt,
      hatch angle/.store in=\hatchangle,hatch angle=0,
      hatch linewidth/.store in=\hatchlinewidth,hatch linewidth=.4pt,
    },
    code={
        \draw[line width=\hatchlinewidth] (0,0) -- (\hatchsize,\hatchsize);
    }
  }

  \node[state,initial above,initstate] (r) {$\{r\}$};
  \node[state, right of=r,node distance=25mm] (rst) {$\{r,s,t\}$};
  \node[state, right of=rst,node distance=25mm] (st) {$\{s,t\}$};
  \node[state, accepting, above of=rst] (em) {$\emptyset$};

  \node[state, accepting, left of=r, node distance=32mm, yshift=-20, draw=black, text=black] (r1) {$\big(\{r{:}3, s{:}0, t{:}1\}, \emptyset\big)$};
  \node[state, accepting, left of=r, node distance=32mm, yshift=10, draw=black, text=black] (r2) {$\big(\{r{:}3, s{:}2, t{:}1\}, \emptyset\big)$};
  \node[state, accepting, left of=r, node distance=32mm, yshift=40, draw=black, text=black] (r3) {$\big(\{r{:}1, s{:}2, t{:}3\}, \emptyset\big)$};

  \node[state, accepting, below of=r,xshift=0mm] (r4) {$\big(\{r{:}1, s{:}0, t{:}0\}, \emptyset\big)$};
  \node[state, accepting, right of=r4,xshift=14mm] (r7) {$\big(\{r{:}1, s{:}0, t{:}1\}, \emptyset\big)$};

  \node[state, accepting, below of=r4, node distance=14mm] (r5) {$\big(\{r{:}1\}, \emptyset\big)$};
  \node[state, below of=r5, node distance=14mm] (r6) {$\big(\{r{:}1,s{:}0, t{:}0\}, \{ s,t \}\big)$};

  \node[state, accepting, right of=r7,xshift=10mm, draw=black, text=black] (r8)  {$\big(\{s{:}0, t{:}1\}, \emptyset\big)$};



  \path (r) edge[bend left]  node[above] {$b$} (r1)
        (r) edge  node[above] {$b$} (r2)
        (r) edge[bend right]  node[above] {$b$} (r3)
        (r) edge  node[left] {$b$} (r4)
        (r) edge[bend right=5]  node[above] {$b$} (r7)
        (r) edge[bend left]  node[above] {$b$} (rst)
        (r) edge[bend left]  node[above] {$a$} (em)
        (rst) edge  node[above] {$a$} (r8)
        (rst) edge[bend left=20]  node[above] {$b$} (r)
        (rst) edge  node[above] {$a$} (st)
        (st) edge  node[above] {$a$} (r8)
        (st) edge [loop above] node {$a$} (st)
        (st) edge [bend right,above] node {$b$} (em)
        (em) edge [loop above] node {$a,b$} (em)
        (r4) edge  node[left] {$b$} (r5)
        (r7) edge[bend left]  node[above] {$b$} (r5)
        (r5) edge[bend left] node[right] {$b$} (r6)
        (r6) edge[bend left] node[left] {$b$} (r5);

  %
  %
  \begin{pgfonlayer}{background}
  \node[rectangle,rounded corners=8pt, opacity=0.1,pattern={hatch[hatch size= 8pt, hatch linewidth=3pt, hatch angle=10]},inner sep=5pt,fit=(r1) (r2) (r3)] (useless1) {};
  \node[rectangle,rounded corners=8pt, opacity=0.1,pattern={hatch[hatch size= 8pt, hatch linewidth=3pt, hatch angle=10]},inner sep=5pt,fit=(r8)] (useless2) {};
  \end{pgfonlayer}
  %

\end{tikzpicture}
  		}
  		\caption{A part of $\algschewe(\aut)$}
  		\label{fig:schewe-res}
  	\end{subfigure}
  \end{minipage}
}
\vspace{-2mm}
\caption{Schewe's complementation}
\label{fig:schewe-ex}
\end{wrapfigure}
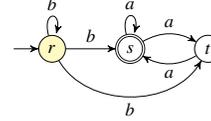
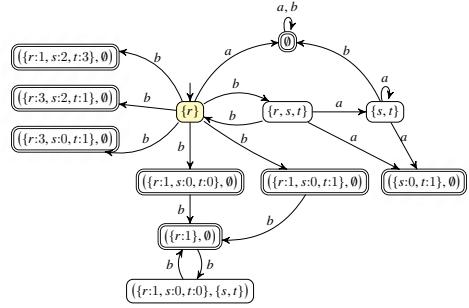%
}


The space complexity of Schewe's construction for BAs matches the theoretical
lower bound $\bigO((0.76n)^n)$ given by Yan~\cite{yan} modulo a quadratic factor
$\bigO(n^2)$. Note that our extension to BAs with accepting transitions does not
increase the space complexity of the construction.

\begin{wrapfigure}[17]{r}{6.3cm}
\vspace*{-14mm}
\hspace*{-3mm}
\scalebox{0.96}{
  \begin{minipage}{6.5cm}
  	\centering
  	\begin{subfigure}[b]{1.0\textwidth}
  		\centering
  		\scalebox{0.9}{
  			\begin{tikzpicture}[->,>=stealth',shorten >=0pt,auto,node distance=1.5cm,
                    scale=0.8,transform shape,initial text={}]
  \tikzstyle{every state}=[inner sep=3pt,minimum size=5pt]
  \tikzstyle{empty}=[]
  \tikzstyle{initstate}=[fill=yellow!30]

  \node[state,initial,initstate] (r) {$r$};
  \node[state,accepting,right of=r] (s) {$s$};
  \node[state,right of=s] (t) {$t$};

  \path (r) edge[loop above]  node {$b$} (r)
        (r) edge  node {$b$} (s)
        (r) edge[bend right=60] node[below] {$b$} (t)
        (s) edge[loop above]  node {$a$} (s)
        (s) edge[bend left]  node {$a$} (t)
        (t) edge[bend left]  node {$a$} (s);

\end{tikzpicture}
  		}
  		\caption{BA $\aut$ over $\{ a,b \}$}
  		\label{fig:schewe-aut}
  	\end{subfigure}
  	~
  	\begin{subfigure}[b]{1.0\textwidth}
  		\centering
  		\scalebox{0.9}{
  			\begin{tikzpicture}[->,>=stealth',shorten >=0pt,auto,node distance=1.8cm,
                    scale=0.6,transform shape,initial text={}]
  \tikzstyle{every state}=[inner sep=3pt,minimum size=5pt,rectangle,rounded corners=1mm]
  \tikzstyle{empty}=[]
  \tikzstyle{initstate}=[fill=yellow!30]
  \tikzstyle{wobbly}=[decorate, decoration={snake,amplitude=.2mm,segment length=2mm,post length=1mm}]

  \tikzdeclarepattern{
    name=hatch,
    parameters={\hatchsize,\hatchangle,\hatchlinewidth},
    bounding box={(-.1pt,-.1pt) and (\hatchsize+.1pt,\hatchsize+.1pt)},
    tile size={(\hatchsize,\hatchsize)},
    tile transformation={rotate=\hatchangle},
    defaults={
      hatch size/.store in=\hatchsize,hatch size=5pt,
      hatch angle/.store in=\hatchangle,hatch angle=0,
      hatch linewidth/.store in=\hatchlinewidth,hatch linewidth=.4pt,
    },
    code={
        \draw[line width=\hatchlinewidth] (0,0) -- (\hatchsize,\hatchsize);
    }
  }

  \node[state,initial above,initstate] (r) {$\{r\}$};
  \node[state, right of=r,node distance=25mm] (rst) {$\{r,s,t\}$};
  \node[state, right of=rst,node distance=25mm] (st) {$\{s,t\}$};
  \node[state, accepting, above of=rst] (em) {$\emptyset$};

  \node[state, accepting, left of=r, node distance=32mm, yshift=-20, draw=black, text=black] (r1) {$\big(\{r{:}3, s{:}0, t{:}1\}, \emptyset\big)$};
  \node[state, accepting, left of=r, node distance=32mm, yshift=10, draw=black, text=black] (r2) {$\big(\{r{:}3, s{:}2, t{:}1\}, \emptyset\big)$};
  \node[state, accepting, left of=r, node distance=32mm, yshift=40, draw=black, text=black] (r3) {$\big(\{r{:}1, s{:}2, t{:}3\}, \emptyset\big)$};

  \node[state, accepting, below of=r,xshift=0mm] (r4) {$\big(\{r{:}1, s{:}0, t{:}0\}, \emptyset\big)$};
  \node[state, accepting, right of=r4,xshift=14mm] (r7) {$\big(\{r{:}1, s{:}0, t{:}1\}, \emptyset\big)$};

  \node[state, accepting, below of=r4, node distance=14mm] (r5) {$\big(\{r{:}1\}, \emptyset\big)$};
  \node[state, below of=r5, node distance=14mm] (r6) {$\big(\{r{:}1,s{:}0, t{:}0\}, \{ s,t \}\big)$};

  \node[state, accepting, right of=r7,xshift=10mm, draw=black, text=black] (r8)  {$\big(\{s{:}0, t{:}1\}, \emptyset\big)$};



  \path (r) edge[bend left]  node[above] {$b$} (r1)
        (r) edge  node[above] {$b$} (r2)
        (r) edge[bend right]  node[above] {$b$} (r3)
        (r) edge  node[left] {$b$} (r4)
        (r) edge[bend right=5]  node[above] {$b$} (r7)
        (r) edge[bend left]  node[above] {$b$} (rst)
        (r) edge[bend left]  node[above] {$a$} (em)
        (rst) edge  node[above] {$a$} (r8)
        (rst) edge[bend left=20]  node[above] {$b$} (r)
        (rst) edge  node[above] {$a$} (st)
        (st) edge  node[above] {$a$} (r8)
        (st) edge [loop above] node {$a$} (st)
        (st) edge [bend right,above] node {$b$} (em)
        (em) edge [loop above] node {$a,b$} (em)
        (r4) edge  node[left] {$b$} (r5)
        (r7) edge[bend left]  node[above] {$b$} (r5)
        (r5) edge[bend left] node[right] {$b$} (r6)
        (r6) edge[bend left] node[left] {$b$} (r5);

  %
  %
  \begin{pgfonlayer}{background}
  \node[rectangle,rounded corners=8pt, opacity=0.1,pattern={hatch[hatch size= 8pt, hatch linewidth=3pt, hatch angle=10]},inner sep=5pt,fit=(r1) (r2) (r3)] (useless1) {};
  \node[rectangle,rounded corners=8pt, opacity=0.1,pattern={hatch[hatch size= 8pt, hatch linewidth=3pt, hatch angle=10]},inner sep=5pt,fit=(r8)] (useless2) {};
  \end{pgfonlayer}
  %

\end{tikzpicture}
  		}
  		\caption{A part of $\algschewe(\aut)$}
  		\label{fig:schewe-res}
  	\end{subfigure}
  \end{minipage}
}
\vspace{-2mm}
\caption{Schewe's complementation}
\label{fig:schewe-ex}
\end{wrapfigure}%

\beginexample\label{ex:schewe}
Consider the BA $\aut$ over $\{ a, b \}$ given in
\cref{fig:schewe-aut}.
A~part of $\algschewe(\aut)$ is shown in \cref{fig:schewe-res}
(we use $(\{ s{:}0, t{:}1 \}, \emptyset)$ to denote the macrostate
$(\{ s, t\}, \emptyset, \{ s\mapsto 0, t\mapsto 1 \}, 0)$).
We omit the $i$-part of each macrostate since the
corresponding values are~0 for all macrostates in the figure.
Useless states are covered by grey stripes.
The full automaton contains even more transitions from
$\{r\}$ to useless macrostates of the form $(\{ r{:}\cdot, s{:}\cdot, t{:}\cdot \},
\emptyset)$.
\qed



From the construction of $\algschewe(\aut)$, we can see that the number of states
is affected mainly by sizes of macrostates and by the maximum rank of $\aut$. In
particular, the upper bound on the number of states of the complement with the
maximum rank $r$ is given in the following lemma.



\begin{restatable}{lemma}{theLemElevBound}\label{lem:rank-bound}
  For a BA $\aut$ with sufficiently many states~$n$ such that $\rankof \aut = r$
  the number of states of the complemented automaton is
	bounded by $2^n + \frac{(r+m)^n}{(r+m)!}$ where $m = \max\{0, 3 - \lceil
	\frac{r}{2} \rceil\}$.
\end{restatable}
From \cref{lemma:maxRank} we have that the rank of $\aut$ is bounded by $2|Q|$.
Such a~bound is often too coarse and hence $\algschewe(\aut)$ may
contain many redundant states.
Decreasing the bound on the ranks is essential for a~practical algorithm, but an
optimal solution is \pspace-complete~\cite{GurumurthyKSV03}.
The rest of this paper therefore proposes
a~framework of lightweight techniques for decreasing the maximum rank bound and,
in this way, significantly reducing the size of the complemented BA.

\vspace{-3.0mm}
\subsection{Tight Rank Upper Bounds}
\vspace{-1.0mm}

Let $\alpha \notin \langof{\aut}$.
For $\ell \in \omega$, we define the $\ell$-th \emph{level} of $\dagw$ as
$\levelw(\ell) = \{ q \mid (q,\ell)\in\dagw \}$. Furthermore, we use
$f^\alpha_\ell$ to denote the ranking of level $\ell$ of $\dagw$.
Formally,
\begin{equation}
f^\alpha_\ell(q) =
\begin{cases}
  \rankwof{(q,\ell)} & \text{if } q\in\levelwof{\ell},\\
   0 & \text{otherwise.}
\end{cases}
\end{equation}
We say that the $\ell$-th level of~$\dagw$ is \emph{tight} if for all $k\geq \ell$ it holds that
\begin{inparaenum}[(i)]
	\item $f^\alpha_k$~is tight, and
	\item $\rankof{f^\alpha_k} = \rankof{f^\alpha_\ell}$.
\end{inparaenum}
Let $\rho = S_0S_1\dots S_{\ell-1}(S_\ell, O_\ell, f_\ell, i_\ell)\dots$ be a run on a
word $\alpha$ in $\algschewe(\aut)$. We say that $\rho$ is a \emph{super-tight
run}~\cite{HavlenaL2021} if $f_k = f^\alpha_k$ for each $k \geq \ell$.
Finally, we say that a~mapping $\mu\colon 2^Q \to \cR$ is a~\emph{tight rank upper bound
(\trub) wrt~$\alpha$}~iff
\begin{equation}\label{eq:trub}
	\exists \ell\in\omega\colon \levelw(\ell)
	\text{ is tight} \land (\forall k \geq \ell \colon \mu(\levelw(k)) \geq
  f^\alpha_k).
\end{equation}
Informally, a~\trub is a~ranking that gives a~conservative (i.e., larger)
estimate on the necessary ranks of states in a~super-tight run. We say that
$\mu$~is a~\trub iff~$\mu$ is a~\trub wrt all $\alpha\notin\langof{\aut}$.
We abuse notation and use the term \trub also for a~mapping $\mu'\colon 2^Q \to
\omega$ if the mapping~$\infunc(\mu')$
is a~\trub where $\infunc(\mu')(S) = \{ q \mapsto m \mid m = \mu'(S)
\monus 1 \text{ if } q \in \accstates \text{ else } m = \mu'(S) \}$ for all $S\in 2^Q$.
($\monus$ is the \emph{monus} operator, i.e., minus with negative results
saturated to zero.)
Note that the mappings $\mu_t = \{S \mapsto (2|S\setminus
\accstates| \monus 1)\}_{S \in 2^Q}$ and $\infunc(\mu_t)$ are trivial \trub{s}.

The following lemma shows that we can remove from $\algschewe(\aut)$ macrostates
whose ranking is not covered by a~\trub (in particular, we show that the reduced
automaton preserves super-tight runs).

\begin{restatable}{lemma}{theTrubRed}\label{lem:theTrubRed}
	Let $\mu$ be a~\trub and~$\but$ be a~BA
  obtained from $\algschewe(\aut)$ by replacing all occurrences of~$Q_2$ by
  $Q_2' = \{ \sofi \mid f \leq \mu(S) \}$.
  Then, $\langof{\but} = \Sigma^\omega \setminus \langof{\aut}$.
\end{restatable}

\newcommand{\figelevrules}{
\begin{figure}[t]
	%
  \begin{subfigure}[b]{0.32\linewidth}
		\centering
		\scalebox{0.8}{
	\begin{tikzpicture}[level distance=1.2cm,
              level 1/.style={sibling distance=1.3cm},>=stealth',->]
			\tikzset{every node/.style = {rectangle,draw,rounded corners}}
	    \node (r) [label=90:{$\ell = \max\{ \ell_D, \ell_N + 1, \ell_W \}$},label=180:{$C\colon$}] {$\mathsf{IWA}{:}\ell$}
	        child { node {$\mathsf{D}{:}\ell_D$} }
	        child { node {$\mathsf{N}{:}\ell_N$} }
	        child { node {$\mathsf{IWA}{:}\ell_W$} };
	\end{tikzpicture}}
  \caption{$C$ is $\mathsf{IWA}$}
  \label{fig:elev-iwa}
  \end{subfigure}
  \begin{subfigure}[b]{0.3\linewidth}
		\centering
		\scalebox{0.8}{
	\begin{tikzpicture}[level distance=1.2cm,
              level 1/.style={sibling distance=1.3cm},>=stealth',->]
			\tikzset{every node/.style = {rectangle,draw,rounded corners}}
	    \node (r) [label=90:{$\ell = \max\{ \ell_D + \textdetcircle, \ell_N + 1, \ell_W + \textiwcircle, 2 \}$},label=180:{$C\colon$}] {$\mathsf{D}{:}\ell$}
	        child[draw=orange!80!black,thick] { node[draw=black,thin] {$\mathsf{D}{:}\ell_D$} edge from parent \detcircle }
	        child { node {$\mathsf{N}{:}\ell_N$} }
	        child[draw=blue!80!black,thick] { node[draw=black,thin] {$\mathsf{IWA}{:}\ell_W$} edge from parent \iwcircle };
	\end{tikzpicture}}
  \caption{$C$ is $\mathsf{D}$}
  \label{fig:elev-det}
  \end{subfigure}
	%
  \begin{subfigure}[b]{0.34\linewidth}
		\centering
		\scalebox{0.8}{
      \begin{tikzpicture}[level distance=1.2cm,
              level 1/.style={sibling distance=1.3cm},>=stealth',->]
      		\tikzset{every node/.style = {rectangle,draw,rounded corners}}
          \node (r) [label=90:{$\ell = \max\{ \ell_D + 1, \ell_N, \ell_W + 1 \}$},label=180:{$C\colon$}] {$\mathsf{N}{:}\ell$}
              child { node {$\mathsf{D}{:}\ell_D$} }
              child { node {$\mathsf{N}{:}\ell_N$} }
              child { node {$\mathsf{IWA}{:}\ell_W$} };
      \end{tikzpicture}}
    \caption{$C$ is $\mathsf{N}$}
    \label{fig:elev-nondet}
  \end{subfigure}
  \caption[Rank assignment rules]{Rules for assigning types and rank bounds to
    MSCCs.
    The symbols~\textdetcircle{} and~\textiwcircle{} are interpeted as~$0$ if all the corresponding
    edges from the components having rank~$\ell_D$ and~$\ell_W$, respectively,
    are deterministic; otherwise they are interpreted as~$2$.
    Transitions between two components $C_1$ and $C_2$ are deterministic if the
    BA $(C, \restrof \delta C, \emptyset, \emptyset)$ is deterministic for $C =
    \delta(C_1,\Sigma)\cap (C_1 \cup C_2)$.
    }
	\label{fig:el-rules}
	\vspace*{-4mm}
\end{figure}
}

\newcommand{\figelevstruct}[0]{
\begin{wrapfigure}[5]{r}{4cm}
\vspace*{-4mm}
\hspace*{0mm}
\begin{minipage}{4cm}
	\centering
	\scalebox{0.8}{
	\begin{tikzpicture}[level distance=1.2cm,
              level 1/.style={sibling distance=1.4cm},>=stealth',->]
			\tikzset{every node/.style = {rectangle,draw,rounded corners}}
	    \node (r) [label=180:{$C\colon$}] {$t{:}\ell$}
	        child { node {$\mathsf{D}{:}\ell_D$} }
	        child { node {$\mathsf{N}{:}\ell_N$} }
	        child { node {$\mathsf{IWA}{:}\ell_W$} };
	\end{tikzpicture}}
\end{minipage}
\vspace{-3mm}
\caption{Structure of elevator ranking rules}
\label{fig:elev-rules-struct}
\end{wrapfigure}%
}

\section{Elevator Automata}
\label{sec:elevator}

In this section, we introduce \emph{elevator automata}, which are BAs having
a~particular structure that can be exploited for complementation and
semi-de\-ter\-min\-ization; elevator
automata can be complemented in $\bigO(16^n)$ (cf.\
\cref{lem:elevator_complement_bound}) space instead of $2^{\bigOof{n \log n}}$,
which is the lower bound for unrestricted BAs, and semi-determinized in
$\bigO(2^n)$ instead of $\bigO(4^n)$ (cf.\ 
\ifTR
\cref{sec:elev-semidet}).
\fi
The class of elevator automata is quite general: it can be seen as a~substantial
generalization of semi-deterministic BAs
(SDBAs)~\cite{CourcoubetisY88,BlahoudekHSST16}. Intuitively, an elevator
automaton is a BA whose strongly connected components are all either deterministic or
inherently weak.

Let $\aut = (Q, \trans, I, \accstates \cup \acctrans)$.
$C \subseteq Q$ is
a~\emph{strongly connected component} (SCC) of~$\aut$
if for any pair of states $q, q' \in C$ it holds that~$q$ is reachable
from~$q'$ and~$q'$ is reachable from~$q$.
$C$~is \emph{maximal} (MSCC) if it is not a~proper subset of another SCC.
An MSCC~$C$ is \emph{trivial} iff $|C| = 1$ and $\restrof \trans C = \emptyset$.
The~\emph{condensation} of~$\aut$ is the DAG $\condensof \aut = (\M, \E)$
where~$\M$ is the set of~$\aut$'s MSCCs and $\E = \{(C_1, C_2) \mid \exists q_1
\in C_1, \exists q_2 \in C_2, \exists a \in \Sigma\colon q_1 \ltr a q_2 \in
\trans\}$.
An MSCC is \emph{non-accepting} if it contains no accepting state and no
accepting transition, i.e., $C \cap \accstates = \emptyset$ and
$\restrof{\trans}{C} \cap \acctrans = \emptyset$.
The \emph{depth} of $(\M, \E)$
is defined as the number of MSCCs on the longest path in $(\M, \E)$.

We say that an SCC~$C$ is \emph{inherently weak accepting} (IWA) iff \emph{every
cycle} in the transition diagram of~$\aut$ restricted to~$C$ contains an
accepting state or an accepting transition.
$C$~is \emph{inherently weak} if it is either non-accepting or IWA, and~$\aut$ is
inherently weak if all of its MSCCs are inherently weak.
%
%
$\aut$ is \emph{deterministic} iff $|I| \leq 1$ and $|\trans(q, a)| \leq 1$ for
all $q\in Q$ and $a \in \Sigma$.
An~SCC~$C \subseteq Q$ is \emph{deterministic} iff $(C, \restrof \trans C,
\emptyset, \emptyset)$ is deterministic.
$\aut$~is a~\emph{semi-deterministic BA} (SDBA) if $\aut[q]$ is deterministic
for every $q \in \accstates \cup \{p \in Q \mid s
\ltr a p \in \acctrans, s \in Q, a \in \Sigma\}$, i.e., whenever a~run
in~$\aut$ reaches an accepting state or an accepting transition, it can only continue
deterministically.

\newcommand{\figltlelevexample}[0]{
\begin{wrapfigure}[14]{r}{5cm}
\vspace*{2mm}
\hspace*{0mm}
\begin{minipage}{5cm}
	\centering
	\scalebox{0.73}{
	\begin{tikzpicture}[->,>=stealth,scale=1,transform shape]
\tikzstyle{every state}=[inner sep=3pt,minimum size=5pt]

  \node[state,  accepting, initial text={}] (q0) {$0$};
  \node[state, node distance=20mm, right of=q0] (q1) {$1$};
  \node[state, accepting, node distance=20mm, below of=q0] (q2) {$2$};
  \node[state, node distance=20mm, right of=q2] (q3) {$3$};
  \node[state, accepting, node distance=20mm, below of=q2] (q4) {$4$};
  \node[state, node distance=20mm, right of=q4] (q5) {$5$};
  \node[above left of=q0,node distance=9mm] (hid) {};

  \draw (hid) edge (q0);

  \draw
  (q0) edge[bend left] node [scale=0.8, above] {$\neg a$} (q1)
  (q0) edge[loop above] node [scale=0.8, above] {$a$} (q0)
  (q1) edge[loop above] node [scale=0.8, above] {$\neg a$} (q1)
  (q1) edge[bend left] node [scale=0.8, above] {$a$} (q0)

  (q2) edge[bend left] node [scale=0.8, above] {$\neg b$} (q3)
  (q2) edge[loop left] node [scale=0.8, below, pos=0.1] {$b$} (q2)
  (q3) edge[loop right] node [scale=0.8, below, pos=0.9] {$\neg b$} (q3)
  (q3) edge[bend left] node [scale=0.8, above] {$b$} (q2)

  (q4) edge[bend left] node [scale=0.8, above] {$\neg c$} (q5)
  (q4) edge[loop below] node [scale=0.8, below] {$c$} (q4)
  (q5) edge[loop below] node [scale=0.8, below] {$\neg c$} (q5)
  (q5) edge[bend left] node [scale=0.8, above] {$c$} (q4)

  (q0) edge[left] node [scale=0.8, left, pos=0.4] {$\neg a \wedge b$} (q2)
  (q1) edge[right] node [scale=0.8, right, pos=0.4] {$\neg a \wedge b$} (q2)

  (q2) edge[left] node [scale=0.8, left] {$\neg b \wedge c$} (q4)
  (q3) edge[right] node [scale=0.8, right] {$\neg b \wedge c$} (q4)
  ;

  \draw [->] (q0) to [out=180,in=90] ($(q2)-(1,0)$) node[left, scale=0.8, yshift=10mm,xshift=-3mm,rotate=90] {$\neg a \wedge \neg b \wedge c$} to[out=-90, in=180] (q4);
  \draw [->] (q1) to [out=0,in=0] ($(q5)+(0.3,-1)$) node[right, scale=0.8, xshift=20mm,yshift=19mm,rotate=90] {$\neg a \wedge \neg b \wedge c$} to[out=180, in=-60] (q4);

  \begin{pgfonlayer}{background}
    \node[draw,dashed,rectangle,fill=pink!50,draw=black!70,rounded corners=8pt,inner sep=15pt,fit=(q0) (q1), yshift=3mm] () {};
    \node[draw,dashed,rectangle,fill=pink!50,draw=black!70,rounded corners=8pt,inner sep=15pt,fit=(q2) (q3), yshift=0.1mm] () {};
    \node[draw,dashed,rectangle,fill=pink!50,draw=black!70,rounded corners=8pt,inner sep=15pt,fit=(q4) (q5), yshift=-1mm] () {};
  \end{pgfonlayer}

  \node[above of=q0,node distance=8mm, xshift=32mm, yshift=1mm] {\emph{det}};
  \node[above of=q2,node distance=8mm, xshift=32mm, yshift=-1mm] {\emph{det}};
  \node[above of=q4,node distance=8mm, xshift=32mm, yshift=-2mm] {\emph{det}};
\end{tikzpicture}
	}
\end{minipage}
\vspace{-3mm}
\caption{The BA for LTL formula $\lG\lF(a \vee \lG\lF(b \vee \lG\lF c))$ is elevator}
\label{fig:elev-ltl-example}
\end{wrapfigure}%
}

\begin{wrapfigure}[14]{r}{5cm}
\vspace*{2mm}
\hspace*{0mm}
\begin{minipage}{5cm}
	\centering
	\scalebox{0.73}{
	\begin{tikzpicture}[->,>=stealth,scale=1,transform shape]
\tikzstyle{every state}=[inner sep=3pt,minimum size=5pt]

  \node[state,  accepting, initial text={}] (q0) {$0$};
  \node[state, node distance=20mm, right of=q0] (q1) {$1$};
  \node[state, accepting, node distance=20mm, below of=q0] (q2) {$2$};
  \node[state, node distance=20mm, right of=q2] (q3) {$3$};
  \node[state, accepting, node distance=20mm, below of=q2] (q4) {$4$};
  \node[state, node distance=20mm, right of=q4] (q5) {$5$};
  \node[above left of=q0,node distance=9mm] (hid) {};

  \draw (hid) edge (q0);

  \draw
  (q0) edge[bend left] node [scale=0.8, above] {$\neg a$} (q1)
  (q0) edge[loop above] node [scale=0.8, above] {$a$} (q0)
  (q1) edge[loop above] node [scale=0.8, above] {$\neg a$} (q1)
  (q1) edge[bend left] node [scale=0.8, above] {$a$} (q0)

  (q2) edge[bend left] node [scale=0.8, above] {$\neg b$} (q3)
  (q2) edge[loop left] node [scale=0.8, below, pos=0.1] {$b$} (q2)
  (q3) edge[loop right] node [scale=0.8, below, pos=0.9] {$\neg b$} (q3)
  (q3) edge[bend left] node [scale=0.8, above] {$b$} (q2)

  (q4) edge[bend left] node [scale=0.8, above] {$\neg c$} (q5)
  (q4) edge[loop below] node [scale=0.8, below] {$c$} (q4)
  (q5) edge[loop below] node [scale=0.8, below] {$\neg c$} (q5)
  (q5) edge[bend left] node [scale=0.8, above] {$c$} (q4)

  (q0) edge[left] node [scale=0.8, left, pos=0.4] {$\neg a \wedge b$} (q2)
  (q1) edge[right] node [scale=0.8, right, pos=0.4] {$\neg a \wedge b$} (q2)

  (q2) edge[left] node [scale=0.8, left] {$\neg b \wedge c$} (q4)
  (q3) edge[right] node [scale=0.8, right] {$\neg b \wedge c$} (q4)
  ;

  \draw [->] (q0) to [out=180,in=90] ($(q2)-(1,0)$) node[left, scale=0.8, yshift=10mm,xshift=-3mm,rotate=90] {$\neg a \wedge \neg b \wedge c$} to[out=-90, in=180] (q4);
  \draw [->] (q1) to [out=0,in=0] ($(q5)+(0.3,-1)$) node[right, scale=0.8, xshift=20mm,yshift=19mm,rotate=90] {$\neg a \wedge \neg b \wedge c$} to[out=180, in=-60] (q4);

  \begin{pgfonlayer}{background}
    \node[draw,dashed,rectangle,fill=pink!50,draw=black!70,rounded corners=8pt,inner sep=15pt,fit=(q0) (q1), yshift=3mm] () {};
    \node[draw,dashed,rectangle,fill=pink!50,draw=black!70,rounded corners=8pt,inner sep=15pt,fit=(q2) (q3), yshift=0.1mm] () {};
    \node[draw,dashed,rectangle,fill=pink!50,draw=black!70,rounded corners=8pt,inner sep=15pt,fit=(q4) (q5), yshift=-1mm] () {};
  \end{pgfonlayer}

  \node[above of=q0,node distance=8mm, xshift=32mm, yshift=1mm] {\emph{det}};
  \node[above of=q2,node distance=8mm, xshift=32mm, yshift=-1mm] {\emph{det}};
  \node[above of=q4,node distance=8mm, xshift=32mm, yshift=-2mm] {\emph{det}};
\end{tikzpicture}
	}
\end{minipage}
\vspace{-3mm}
\caption{The BA for LTL formula $\lG\lF(a \vee \lG\lF(b \vee \lG\lF c))$ is elevator}
\label{fig:elev-ltl-example}
\end{wrapfigure}%

$\aut$~is an \emph{elevator (B\"{u}chi) automaton} iff for every MSCC~$C$
of~$\aut$ it holds that $C$ is
\begin{inparaenum}[(i)]
  \item  deterministic,
  \item  IWA, or
  \item  non-accepting.
\end{inparaenum}
In other words, a~BA is an elevator automaton iff every nondeterministic SCC of~$\aut$ that
contains an accepting state or transition is inherently weak.
%
An example of an elevator automaton obtained from the LTL formula $\lG\lF(a \vee
\lG\lF(b \vee \lG\lF c))$ is shown in \cref{fig:elev-ltl-example}.
The BA consists of three connected deterministic components.
Note that the automaton is neither semi-deterministic nor unambiguous.

The rank of an elevator automaton~$\aut$ does not depend on the number of
states (as in general BAs), but only on the number of MSCCs and the depth
of~$\condensof \aut$.
In the worst case, $\aut$ consists of a chain of deterministic components,
yielding the upper bound on the rank of elevator automata given in the
following lemma.

\begin{restatable}{lemma}{theLemElevRank}\label{lem:elevator-rank}
	Let~$\aut$ be an elevator automaton such that its condensation has the depth~$d$.
	Then $\rankof \aut \leq 2d$.
\end{restatable}




\vspace{-4.0mm}
\subsection{Refined Ranks for Elevator Automata}\label{sec:ref-ranks-elev}
\vspace{-0.0mm}

Notice that the upper bound on ranks provided by \cref{lem:elevator-rank}
can still be too coarse.
For instance, for an SDBA with three linearly ordered MSCCs such that the first
two MSCCs are non-accepting and the last one is
deterministic accepting, the lemma gives us an upper bound on the rank~$6$, while it is known
that every SDBA has the rank at most~$3$ (cf.~\cite{BlahoudekHSST16}).
Another examples might be two deterministic non-trivial MSCCs connected by a~path of
trivial MSCCs, which can be assigned the same rank.

Instead of refining the definition of elevator automata into some quite complex
list of constraints, we rather provide an algorithm that performs a~traversal
through~$\condensof \aut$ and assigns each MSCC a~label of the form
\scclabel{\emph{type}}{\emph{rank}} that contains
\begin{inparaenum}[(i)]
  \item  a~type and
  \item  a~bound on the maximum rank of states in the component.
\end{inparaenum}
The types of MSCCs that we consider are the following:

\vspace{-2mm}
\begin{description}
  \item[$\mathsf{T}$:] \emph{trivial} components,
  \item[$\mathsf{IWA}$:] \emph{inherently weak accepting} components,
  \item[$\mathsf{D}$:] \emph{deterministic} (potentially accepting) components, and
  \item[$\mathsf{N}$:] \emph{non-accepting} components.
\end{description}
\vspace{-2mm}

Note that the type in an MSCC is not given \emph{a priori} but is
determined by the algorithm (this is because for deterministic non-accepting
components, it is sometimes better to treated them as~$\mathsf{D}$ and
sometimes as~$\mathsf{N}$, depending on their neighbourhood).
In~the following, we assume that~$\aut$ is an elevator automaton without
useless states and, moreover, all accepting conditions on states and
transitions not inside non-trivial MSCCs are removed (any BA can be easily transformed into this
form).

We start with terminal MSCCs~$C$, i.e., MSCCs that cannot reach any other MSCC:
\begin{enumerate}
  \item[\textbf{T1}:]  If $C$~is IWA, then we label it
    with \scclabel{$\mathsf{IWA}$} 0.
  \item[\textbf{T2}:]  Else if $C$~is deterministic accepting, we label it with
    \scclabel{$\mathsf{D}$} 2.
\end{enumerate}

\newcommand{\figelevexample}[0]{
\begin{wrapfigure}[12]{r}{6.2cm}
\vspace*{-8mm}
\hspace*{-2mm}
\begin{minipage}{6.4cm}
	\centering
	\scalebox{0.9}{
	\begin{tikzpicture}[->,>=stealth',shorten >=0pt,auto,node distance=1.8cm,
                    scale=0.6,transform shape,initial text={}]
  \tikzstyle{every state}=[inner sep=3pt,minimum size=5pt,rectangle,rounded corners=1mm]
  \tikzstyle{empty}=[]
  \tikzstyle{initstate}=[fill=yellow!30]
  \tikzstyle{wobbly}=[decorate, decoration={snake,amplitude=.2mm,segment length=2mm,post length=1mm}]

  \node[state,initial above,initstate] (r) {$\{r\}$};
  \node[state, right of=r,node distance=25mm] (rst) {$\{r,s,t\}$};
  \node[state, right of=rst,node distance=25mm] (st) {$\{s,t\}$};
  \node[state, accepting, above of=rst] (em) {$\emptyset$};

  \node[state, accepting, left of=r, node distance=32mm, yshift=-20, draw=black, text=black] (r1) {$\big(\{r{:}3, s{:}0, t{:}1\}, \emptyset\big)$};
  \node[state, accepting, left of=r, node distance=32mm, yshift=10, draw=black, text=black] (r2) {$\big(\{r{:}3, s{:}2, t{:}1\}, \emptyset\big)$};
  \node[state, accepting, left of=r, node distance=32mm, yshift=40, draw=black, text=black] (r3) {$\big(\{r{:}1, s{:}2, t{:}3\}, \emptyset\big)$};

  \node[state, accepting, below of=r,xshift=0mm] (r4) {$\big(\{r{:}1, s{:}0, t{:}0\}, \emptyset\big)$};
  \node[state, accepting, right of=r4,xshift=14mm] (r7) {$\big(\{r{:}1, s{:}0, t{:}1\}, \emptyset\big)$};

  \node[state, accepting, below of=r4, node distance=14mm] (r5) {$\big(\{r{:}1\}, \emptyset\big)$};
  \node[state, below of=r5, node distance=14mm] (r6) {$\big(\{r{:}1,s{:}0, t{:}0\}, \{ s,t \}\big)$};

  \node[state, accepting, right of=r7,xshift=10mm, draw=black, text=black] (r8)  {$\big(\{s{:}0, t{:}1\}, \emptyset\big)$};



  \path (r) edge[bend left]  node[above] {$b$} (r1)
        (r) edge  node[above] {$b$} (r2)
        (r) edge[bend right]  node[above] {$b$} (r3)
        (r) edge  node[left] {$b$} (r4)
        (r) edge[bend right=5]  node[above] {$b$} (r7)
        (r) edge[bend left]  node[above] {$b$} (rst)
        (r) edge[bend left]  node[above] {$a$} (em)
        (rst) edge  node[above] {$a$} (r8)
        (rst) edge[bend left=20]  node[above] {$b$} (r)
        (rst) edge  node[above] {$a$} (st)
        (st) edge  node[above] {$a$} (r8)
        (st) edge [loop above] node {$a$} (st)
        (st) edge [bend right,above] node {$b$} (em)
        (em) edge [loop above] node {$a,b$} (em)
        (r4) edge  node[left] {$b$} (r5)
        (r7) edge[bend left]  node[above] {$b$} (r5)
        (r5) edge[bend left] node[right] {$b$} (r6)
        (r6) edge[bend left] node[left] {$b$} (r5);

  %
  %
  %

  \begin{pgfonlayer}{background}
    \draw[-,dashed,rectangle,fill=Yellow!20,draw=black!70,rounded corners=5pt,inner sep=10pt]
      ([xshift=-4mm]r4.west) |- ([yshift=9mm]em.north) -| ([xshift=3mm,yshift=-3mm]st.south east) -| ([xshift=4mm]r4.east) |-
      ([xshift=-3mm,yshift=-3mm]r6.south) -| ([xshift=-4mm]r4.north west);
  \end{pgfonlayer}

  \begin{pgfonlayer}{background}
  \node[rectangle,rounded corners=8pt, opacity=0.1,pattern={hatch[hatch size= 8pt, hatch linewidth=3pt, hatch angle=10]},inner sep=5pt,fit=(r1) (r2) (r3)] (useless1) {};
  \node[rectangle,rounded corners=8pt, opacity=0.1,pattern={hatch[hatch size= 8pt, hatch linewidth=3pt, hatch angle=10]},inner sep=5pt,fit=(r8)] (useless2) {};
  \end{pgfonlayer}

\end{tikzpicture}
	}
\end{minipage}
\vspace{-3mm}
\caption{A part of $\algschewe(\aut)$. The \trub computed by elevator
  rules is used to prune states outside the yellow area. }
\label{fig:compl-elev}
\end{wrapfigure}%
}

\figelevrules

\figelevstruct

\noindent
(Note that the previous two options are complete due to our requirements on the
structure of~$\aut$.)
When all terminal MSCCs are labelled, we proceed through $\condensof \aut$,
inductively on its structure, and label non-terminal components~$C$ based on
the rules defined below.

\begin{wrapfigure}[12]{r}{6.2cm}
\vspace*{-8mm}
\hspace*{-2mm}
\begin{minipage}{6.4cm}
	\centering
	\scalebox{0.9}{
	\begin{tikzpicture}[->,>=stealth',shorten >=0pt,auto,node distance=1.8cm,
                    scale=0.6,transform shape,initial text={}]
  \tikzstyle{every state}=[inner sep=3pt,minimum size=5pt,rectangle,rounded corners=1mm]
  \tikzstyle{empty}=[]
  \tikzstyle{initstate}=[fill=yellow!30]
  \tikzstyle{wobbly}=[decorate, decoration={snake,amplitude=.2mm,segment length=2mm,post length=1mm}]

  \node[state,initial above,initstate] (r) {$\{r\}$};
  \node[state, right of=r,node distance=25mm] (rst) {$\{r,s,t\}$};
  \node[state, right of=rst,node distance=25mm] (st) {$\{s,t\}$};
  \node[state, accepting, above of=rst] (em) {$\emptyset$};

  \node[state, accepting, left of=r, node distance=32mm, yshift=-20, draw=black, text=black] (r1) {$\big(\{r{:}3, s{:}0, t{:}1\}, \emptyset\big)$};
  \node[state, accepting, left of=r, node distance=32mm, yshift=10, draw=black, text=black] (r2) {$\big(\{r{:}3, s{:}2, t{:}1\}, \emptyset\big)$};
  \node[state, accepting, left of=r, node distance=32mm, yshift=40, draw=black, text=black] (r3) {$\big(\{r{:}1, s{:}2, t{:}3\}, \emptyset\big)$};

  \node[state, accepting, below of=r,xshift=0mm] (r4) {$\big(\{r{:}1, s{:}0, t{:}0\}, \emptyset\big)$};
  \node[state, accepting, right of=r4,xshift=14mm] (r7) {$\big(\{r{:}1, s{:}0, t{:}1\}, \emptyset\big)$};

  \node[state, accepting, below of=r4, node distance=14mm] (r5) {$\big(\{r{:}1\}, \emptyset\big)$};
  \node[state, below of=r5, node distance=14mm] (r6) {$\big(\{r{:}1,s{:}0, t{:}0\}, \{ s,t \}\big)$};

  \node[state, accepting, right of=r7,xshift=10mm, draw=black, text=black] (r8)  {$\big(\{s{:}0, t{:}1\}, \emptyset\big)$};



  \path (r) edge[bend left]  node[above] {$b$} (r1)
        (r) edge  node[above] {$b$} (r2)
        (r) edge[bend right]  node[above] {$b$} (r3)
        (r) edge  node[left] {$b$} (r4)
        (r) edge[bend right=5]  node[above] {$b$} (r7)
        (r) edge[bend left]  node[above] {$b$} (rst)
        (r) edge[bend left]  node[above] {$a$} (em)
        (rst) edge  node[above] {$a$} (r8)
        (rst) edge[bend left=20]  node[above] {$b$} (r)
        (rst) edge  node[above] {$a$} (st)
        (st) edge  node[above] {$a$} (r8)
        (st) edge [loop above] node {$a$} (st)
        (st) edge [bend right,above] node {$b$} (em)
        (em) edge [loop above] node {$a,b$} (em)
        (r4) edge  node[left] {$b$} (r5)
        (r7) edge[bend left]  node[above] {$b$} (r5)
        (r5) edge[bend left] node[right] {$b$} (r6)
        (r6) edge[bend left] node[left] {$b$} (r5);

  %
  %
  %

  \begin{pgfonlayer}{background}
    \draw[-,dashed,rectangle,fill=Yellow!20,draw=black!70,rounded corners=5pt,inner sep=10pt]
      ([xshift=-4mm]r4.west) |- ([yshift=9mm]em.north) -| ([xshift=3mm,yshift=-3mm]st.south east) -| ([xshift=4mm]r4.east) |-
      ([xshift=-3mm,yshift=-3mm]r6.south) -| ([xshift=-4mm]r4.north west);
  \end{pgfonlayer}

  \begin{pgfonlayer}{background}
  \node[rectangle,rounded corners=8pt, opacity=0.1,pattern={hatch[hatch size= 8pt, hatch linewidth=3pt, hatch angle=10]},inner sep=5pt,fit=(r1) (r2) (r3)] (useless1) {};
  \node[rectangle,rounded corners=8pt, opacity=0.1,pattern={hatch[hatch size= 8pt, hatch linewidth=3pt, hatch angle=10]},inner sep=5pt,fit=(r8)] (useless2) {};
  \end{pgfonlayer}

\end{tikzpicture}
	}
\end{minipage}
\vspace{-3mm}
\caption{A part of $\algschewe(\aut)$. The \trub computed by elevator
  rules is used to prune states outside the yellow area. }
\label{fig:compl-elev}
\end{wrapfigure}%

The rules are of the form that uses the structure depicted
in~\cref{fig:elev-rules-struct}, where children nodes denote already processed MSCCs.
In particular, a~child node of the form \scclabel{$k$}{$\ell_k$} denotes an
aggregate node of \emph{all} siblings of the type~$k$ with $\ell_k$ being
the maximum rank of these siblings.
Moreover, we use $\typemax \{\exprof D, \exprof N, \exprof W\}$ to denote the
type~$j \in \{\mathsf{D, N, IWA}\}$ for which $\exprof j = \max\{\exprof D,
\exprof N, \exprof W\}$ where~$e_i$ is an expression containing~$\ell_i$ (if
there are more such types, $j$~is chosen arbitrarily).
The rules for assigning a~type~$t$ and a~rank~$\ell$ to~$C$ are the following:
\begin{enumerate}
  \item[\textbf{I1}:]  If~$C$ is trivial, we set $t = \typemax\{\ell_D, \ell_N,
    \ell_W\}$ and $\ell = \max\{\ell_D, \ell_N, \ell_W\}$.
  \item[\textbf{I2}:]  Else if~$C$ is IWA, we use the rule in \cref{fig:elev-iwa}.
  \item[\textbf{I3}:]  Else if~$C$ is deterministic accepting, we use the rule in
    \cref{fig:elev-det}.
  \item[\textbf{I4}:]  Else if~$C$ is deterministic and non-accepting, we try
    both rules from \cref{fig:elev-det,fig:elev-nondet} and pick the rule that
    gives us a~smaller rank.
  \item[\textbf{I5}:]  Else if $C$~is nondeterministic and non-accepting, we use
    the rule in \cref{fig:elev-nondet}.
\end{enumerate}

\noindent
Then, for every MSCC~$C$ of~$\aut$, we assign each of its states the rank
of~$C$. We use $\chi\colon Q \to \omega$ to denote the rank bounds computed by
the procedure above. 

\begin{restatable}{lemma}{theLemElevCorr}\label{lem:elevator-corr}
	$\chi$ is a \trub.
\end{restatable}



%

Using \cref{lem:theTrubRed}, we can now use~$\chi$ to prune states during the
construction of~$\algschewe(\aut)$, as shown in the following example.

\beginexample
As an example, consider the BA $\aut$ in \cref{fig:schewe-aut}.
The set of MSCCs with their types is given as $\{ \{ r \}{:}\mathsf{N}, \{ s,t
\}{:}\mathsf{IWA} \}$ showing that BA $\aut$ is an elevator. Using the rules
\textbf{T1} and \textbf{I4} we get the \trub $\chi = \{ r{:}1, s{:}0, t{:}0
\}$. The TRUB can be used to prune the generated states as shown in
\cref{fig:compl-elev}.
\qed


\vspace{-3.0mm}
\subsection{Efficient Complementation of Elevator Automata}\label{sec:eff-elev}
\vspace{-0.0mm}

In \cref{sec:ref-ranks-elev} we proposed an algorithm for assigning ranks to
MSCCs of an elevator automaton~$\aut$. The drawback of the algorithm is that the
maximum obtained rank is not bounded by a~constant but by the depth of the
condensation of~$\aut$.
We will, however, show that it is actually possible to change~$\aut$ by at most
doubling the number of states and obtain an elevator BA with the rank at
most~3.

Intuitively, the construction copies every non-trivial MSCC~$C$ with an
accepting state or transition into a~component~$\deelevof C$, copies all
transitions going into states in~$C$ to also go into the corresponding states
in~$\deelevof C$, and, finally, removes all accepting conditions from~$C$.
%
%
Formally, let $\aut = (Q, \delta, I, \accstates \cup \acctrans)$ be a BA. For
$C \subseteq Q$, we use $\deelevof{C}$ to denote a unique copy of $C$, i.e.,
$\deelevof{C} = \{ \deelevof{q} \mid q \in C \} $ s.t. $\deelevof{C} \cap Q =
\emptyset$.
Let $\M$ be the set of MSCCs of~$\aut$.
Then, the \emph{deelevated} BA $\algdeelev(\aut) = (Q', \delta', I',
\accstates' \cup \acctrans')$ is given as follows:
%
\begin{itemize}
	\item $Q' = Q \cup \deelevof{Q}$,
	\item $\delta': Q'\times\Sigma \to 2^{Q'}$ where for $q \in Q$
	\begin{itemize}
		\item $\delta'(q,a) = \delta(q,a) \cup \deelevof{(\delta(q,a))}$ and
		\item $\delta'(\deelevof{q},a) = \deelevof{(\delta(q,a)\cap C)}$ for $q \in C\in\M$;
	\end{itemize}
	\item $I' = I$, and
	\item $\accstates' = \deelevof{\accstates}$ and $\acctrans' = \{ \deelevof{q}\ltr{a}\deelevof{r} \mid q\ltr{a} r \in \acctrans \} \cap \delta'$.
\end{itemize}
It is easy to see that the number of states of the deelevated automaton is
bounded by $2|Q|$. Moreover, if~$\aut$ is elevator, so is $\algdeelev(\aut)$.
The construction preserves the language of~$\aut$, as shown by the following lemma.

\begin{restatable}{lemma}{theLemDeElLang}
	Let $\aut$ be a BA. Then, $\langof{\aut} = \langof{\algdeelev(\aut)}$.
\end{restatable}



\noindent
Moreover, for an elevator automaton~$\aut$, the structure of $\algdeelev(\aut)$ consists of
(after trimming useless states) several non-accepting MSCCs
with copied terminal deterministic or IWA MSCCs.
Therefore, if we apply the algorithm from \cref{sec:ref-ranks-elev} on
$\algdeelev(\aut)$, we get that its rank is bounded by~$3$, which gives the
following upper bound for complementation of elevator automata.

\begin{restatable}{lemma}{theLemElevCompl}\label{lem:elevator_complement_bound}
  Let $\aut$ be an elevator automaton with suffficiently many states $n$. Then
  the language $\Sigma^\omega \setminus \langof \aut$ can be represented by a~BA with at most
  $\bigO(16^n)$ states.
\end{restatable}


\noindent
The complementation through $\algdeelev(\aut)$ gives a better upper
bound than the rank refinement from \cref{sec:ref-ranks-elev} applied directly
on $\aut$, however, based on our experience, complementation through $\algdeelev(\aut)$
behaves worse in many real-world instances.
This poor behaviour is caused by the fact that the
complement of $\algdeelev(\aut)$ can have a~larger $\waiting$ and
macrostates in \tight can have larger $S$-components,
which can yield more generated states (despite the rank bound~3).
It seems that the most promising approach would to be a combination of the
approaches, which we leave for future work.

\vspace{-0.0mm}
\subsection{Refined Ranks for Non-Elevator Automata}\label{sec:label}
\vspace{-0.0mm}

The algorithm from \cref{sec:ref-ranks-elev} computing a \trub for elevator
automata can be extended to compute \trub{}s even for general non-elevator
automata (i.e., BAs with nondeterministic accepting components that are not
inherently weak).
To achieve this generalization, we extend the rules for assigning types
and ranks to MSCCs of elevator automata from \cref{sec:ref-ranks-elev}
to take into account general non-deterministic components.
For this, we add into our collection of MSCC types \emph{general} components
(denoted as $\mathsf{G}$).
Further, we need to extend the rules for terminal components with the following rule:
\begin{enumerate}
  \item[\textbf{T3}:]  Otherwise, we label~$C$ with \scclabel{$\mathsf{G}$}{$2|C\setminus Q_F|$}.
\end{enumerate}

\begin{figure}[t]
  \begin{subfigure}[b]{0.31\linewidth}
		\centering
		\scalebox{0.75}{
	\begin{tikzpicture}[level distance=1.2cm,
              level 1/.style={sibling distance=1.3cm},>=stealth',->]
			\tikzset{every node/.style = {rectangle,draw,rounded corners}}
	    \node (r) [label=90:{$\ell = \max\{ \ell_D, \ell_N + 1, \ell_W, \ell_G \}$},label=180:{$C\colon$}] {$\mathsf{IWA}{:}\ell$}
	        child { node {$\mathsf{D}{:}\ell_D$} }
	        child { node {$\mathsf{N}{:}\ell_N$} }
	        child { node {$\mathsf{IWA}{:}\ell_W$} }
					child { node {$\mathsf{G}{:}\ell_G$}  edge from parent [dashed] };
	\end{tikzpicture}}
  \caption{$C$ is $\mathsf{IWA}$}
  \label{fig:gen-elev-iwa}
  \end{subfigure}
  \hspace*{-2mm}
  \begin{subfigure}[b]{0.34\linewidth}
		\centering
		\scalebox{0.75}{
	\begin{tikzpicture}[level distance=1.2cm,
              level 1/.style={sibling distance=1.3cm},>=stealth',->]
			\tikzset{every node/.style = {rectangle,draw,rounded corners}}
	    \node (r) [label=90:{$\ell = \max\{ \ell_D + \textdetcircle, \ell_N + 1,
      \ell_W + \textiwcircle, \ell_G+2, 2 \}$},label=180:{$C\colon$}] {$\mathsf{D}{:}\ell$}
	        child[draw=orange!80!black,thick] { node[draw=black,thin] {$\mathsf{D}{:}\ell_D$} edge from parent \detcircle }
	        child { node {$\mathsf{N}{:}\ell_N$} }
	        child[draw=blue!80!black,thick] { node[draw=black,thin] {$\mathsf{IWA}{:}\ell_W$} edge from parent \iwcircle }
					child { node {$\mathsf{G}{:}\ell_G$}  edge from parent [dashed] };
	\end{tikzpicture}}
  \caption{$C$ is $\mathsf{D}$}
  \label{fig:gen-elev-det}
  \end{subfigure}
%
	\hspace*{3mm}
  \begin{subfigure}[b]{0.31\linewidth}
		\centering
			\scalebox{0.75}{
      \begin{tikzpicture}[level distance=1.2cm,
              level 1/.style={sibling distance=1.3cm},>=stealth',->]
      		\tikzset{every node/.style = {rectangle,draw,rounded corners}}
          \node (r) [label=90:{$\ell = \max\{ \ell_D + 1, \ell_N, \ell_W + 1, \ell_G + 1 \}$},label=180:{$C\colon$}] {$\mathsf{N}{:}\ell$}
              child { node {$\mathsf{D}{:}\ell_D$} }
              child { node {$\mathsf{N}{:}\ell_N$} }
              child { node {$\mathsf{IWA}{:}\ell_W$} }
							child { node {$\mathsf{G}{:}\ell_G$}  edge from parent [dashed] };
      \end{tikzpicture}}
    \caption{$C$ is $\mathsf{N}$}
    \label{fig:gen-elev-nondet}
  \end{subfigure}
  \caption[Rank assignment rules]{Rules assigning types and rank bounds for
  non-elevator automata.
    }
	\label{fig:gen-el-rules}
\end{figure}

\newcommand{\figgenrule}[0]{
\begin{wrapfigure}[4]{r}{5.0cm}
\vspace*{-16mm}
\hspace*{-2mm}
\begin{minipage}{5.0cm}
	\centering
		\scalebox{0.78}{
		\begin{tikzpicture}[level distance=1.2cm,
						level 1/.style={sibling distance=1.3cm},>=stealth',->]
				\tikzset{every node/.style = {rectangle,draw,rounded corners}}
				\node (r) [label=90:{$\ell = \max\{ \ell_D, \ell_N + 1, \ell_W, \ell_G
				\} + 2|C\setminus \accstates|$},label=180:{$C\colon$}] {$\mathsf{G}{:}\ell$}
						child { node {$\mathsf{D}{:}\ell_D$} }
						child { node {$\mathsf{N}{:}\ell_N$} }
						child { node {$\mathsf{IWA}{:}\ell_W$} }
						child { node {$\mathsf{G}{:}\ell_G$}  edge from parent [dashed] };
		\end{tikzpicture}}
\end{minipage}
\vspace{-0mm}
\caption{$C$ is $\mathsf{G}$}
\label{fig:gen-elev-gen}
\end{wrapfigure}%
}

\figgenrule
Moreover, we adjust the rules for assigning a~type~$t$ and a~rank~$\ell$ to~$C$
to the following (the rule \textbf{I1} is the same as for the case of elevator
automata):
\begin{enumerate}
	\item[\textbf{I2}--\textbf{I5}:] \emph{(We replace the corresponding rules for their
		counterparts including general components from \cref{fig:gen-el-rules}).}
  \item[\textbf{I6}:] Otherwise, we use the rule in \cref{fig:gen-elev-gen}.
\end{enumerate}


\noindent
Then, for every MSCC~$C$ of a~BA~$\aut$, we assign each of its states the rank
of~$C$. Again, we use $\chi\colon Q \to \omega$ to denote the rank bounds
computed by the adjusted procedure above.

\begin{restatable}{lemma}{theLemGenElevCorr}\label{lem:gen-elevator-corr}
	$\chi$ is a \trub.
\end{restatable}
%

\tikzstyle{hackennode}=[draw,circle,fill=white,inner sep=0,minimum size=4pt]
\tikzstyle{hackenline}=[line width=3pt]

\newcommand{\figrankprop}[0]{
\begin{wrapfigure}[6]{r}{4.6cm}
\vspace*{-7mm}
\hspace*{0mm}
\begin{minipage}{4.6cm}
	\centering
	\scalebox{0.8}{
	\begin{tikzpicture}[level distance=1.4cm,
              level 1/.style={sibling distance=1.4cm},>=stealth',<-]
			\tikzset{every node/.style = {rectangle,draw,rounded corners}}
	    \node (r) [] {$\mu'(S)$} [grow'=up]
	        child { node {$\mu(R_1)$} edge from parent node[left,xshift=-1mm,draw=none] {$a_1$} }
	        child { node {$\mu(R_2)$} edge from parent node[right,yshift=0mm,draw=none] {$a_2$}  }
	        child { node[draw=none,xshift=-3mm] {$\cdots$} edge from parent[draw=none]}
	        child { node[xshift=-6mm] {$\mu(R_m)$} edge from parent node[right,xshift=1mm,draw=none] {$a_m$}  };
	\end{tikzpicture}}
\end{minipage}
\vspace{-3mm}
\caption{Rank propagation flow}
\label{fig:rank-prop}
\end{wrapfigure}
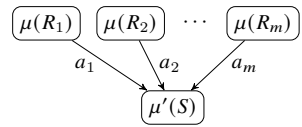%
}
\vspace{-4.0mm}
\section{Rank Propagation}
\label{sec:rank-propagation}
\vspace{-2.0mm}

\figrankprop

In the previous section, we proposed a way, how to obtain a \trub for elevator
automata (with generalization to general automata).
In this section, we propose a~way of using the structure of~$\aut$ to refine
a~\trub using a propagation of values and thus reduce the size of~\tight.
Our approach uses \emph{data flow analysis}~\cite{NielsenNH99} to reason on how
ranks and rankings of macrostates of $\algschewe(\aut)$ can be decreased based
on the ranks and rankings of the \emph{local neighbourhood} of the macrostates.
We, in particular, use a~special case of \emph{forward analysis} working on the
\emph{skeleton} of $\algschewe(\aut)$, which is defined as the BA $\skeletof
\aut = (2^Q, \delta', \emptyset, \emptyset)$ where $\delta' = \{R \ltr a S \mid
S = \delta(R, a)\}$ (note that we are only interested in the structure of
$\skeletof \aut$ and not its language; also notice the similarity of~$\skeletof
\aut$ with \waiting).
Our analysis refines a~rank/ranking estimate~$\mu(S)$ for a~macrostate~$S$
of~$\skeletof \aut$ based on the estimates for its predecessors $R_1, \ldots,
R_m$ (see \cref{fig:rank-prop}).
The new estimate is denoted as~$\mu'(S)$.

More precisely, $\mu\colon 2^Q \to \val$ is a~function giving each macrostate
of~$\skeletof \aut$ a~value from the domain~$\val$.
We will use the following two value domains:
\begin{inparaenum}[(i)]
  \item  $\val = \omega$, which is used for estimating \emph{ranks} of
    macrostates (in the \emph{outer macrostate analysis}) and
  \item  $\val = \cR$, which is used for estimating \emph{rankings} within
    macrostates (in the \emph{inner macrostate analysis}).
\end{inparaenum}
For each of the analyses, we will give the \emph{update function}~$\upd\colon (2^Q
\to \val) \times (2^Q)^{m+1} \to \val$, which defines how the value of~$\mu(S)$
is updated based on the values of $\mu(R_1), \ldots, \mu(R_m)$.
We then construct a~system with the following equation for every $S \in 2^Q$:
\begin{equation}
\mu(S) = \upd(\mu, S, R_1, \ldots, R_m) \quad \text{where } \{R_1, \ldots, R_m\} =
  \delta'^{-1}(S, \Sigma).
\end{equation}
We then solve the system of equations using standard algorithms for data flow
analysis (see, e.g., \cite[Chapter~2]{NielsenNH99}) to obtain the fixpoint
$\mu^*$.
Our analyses have the important property that if they start with~$\mu_0$ being
a~\trub, then~$\mu^*$ will also be a~\trub.

As the initial \trub, we can use
a trivial \trub
or any other \trub (e.g., the output of elevator state
analysis from \cref{sec:elevator}).

%
%
%
%
%

\newcommand{\figouterflow}[0]{
\begin{wrapfigure}[10]{r}{6cm}
  \vspace{-0mm}
  \scalebox{0.96}{
      \begin{minipage}{6cm}
    	\begin{subfigure}{2.3cm}
    		\hspace*{-2mm}
    		\begin{tikzpicture}[->,>=stealth',shorten >=0pt,auto,node distance=1.5cm,
                    scale=0.6,transform shape,initial text={}]
  \tikzstyle{every state}=[inner sep=3pt,minimum size=5pt]
  \tikzstyle{empty}=[]
  \tikzstyle{initstate}=[fill=yellow!30]

  \node[state,initial,initstate] (p) {$p$};
  \node[state,right of=p] (q) {$q$};
  \node[state,right of=q] (r) {$r$};
  \node[state,below of=q] (s) {$s$};

  \path (p) edge[loop below]  node {$a$} (p)
        (p) edge  node {$a$} (q)
        (q) edge  node {$a$} node[below] {$\bullet$} (r)
        (q) edge  node {$a$} (s)
        (r) edge[bend right=60]  node[above] {$a$} (p)
        (s) edge  node[right,yshift=2mm,xshift=-1mm] {$a$} node[left,yshift=-2mm,xshift=1mm] {$\bullet$} (p);

\end{tikzpicture}
    		\caption{$\autex$}
    		\label{fig:flow-aut}
    	\end{subfigure}
    	~
    	\begin{subfigure}{1.6cm}
    		\centering
    		\begin{tikzpicture}[->,>=stealth',shorten >=0pt,auto,node distance=1.8cm,
                    scale=0.6,transform shape,initial text={}]
  \tikzstyle{every state}=[inner sep=3pt,minimum size=5pt,rectangle,rounded corners=1mm]
  \tikzstyle{empty}=[]
  \tikzstyle{initstate}=[fill=yellow!30]
  \tikzstyle{wobbly}=[decorate, decoration={snake,amplitude=.2mm,segment length=2mm,post length=1mm}]

  \node[state,initial,initstate,label=0:{${:}1$}] (r) {$\{p\}$};
  \node[state, below of=r,node distance=12mm,label=0:{${:}3$}] (rst) {$\{p,q\}$};
  \node[state, below of=rst,node distance=12mm,label=0:{${:}7$}] (st) {$\{p,q,r,s\}$};

  \path (r) edge node[left] {$a$} (rst)
        (rst) edge node[left] {$a$} (st)
        (st) edge[loop below] node {$a$} (st);

\end{tikzpicture}
    		\caption{$\mu_0$}
    		\label{fig:flow-init}
    	\end{subfigure}
      ~
    	\begin{subfigure}{1.6cm}
    		\centering
    		\begin{tikzpicture}[->,>=stealth',shorten >=0pt,auto,node distance=1.8cm,
                    scale=0.6,transform shape,initial text={}]
  \tikzstyle{every state}=[inner sep=3pt,minimum size=5pt,rectangle,rounded corners=1mm]
  \tikzstyle{empty}=[]
  \tikzstyle{initstate}=[fill=yellow!30]
  \tikzstyle{wobbly}=[decorate, decoration={snake,amplitude=.2mm,segment length=2mm,post length=1mm}]

  \node[state,initial,initstate,label=0:{${:}1$}] (r) {$\{p\}$};
  \node[state, below of=r,node distance=12mm,label=0:{${:}\textcolor{red}{1}$}] (rst) {$\{p,q\}$};
  \node[state, below of=rst,node distance=12mm,label=0:{${:}7$}] (st) {$\{p,q,r,s\}$};

  \path (r) edge node[left] {$a$} (rst)
        (rst) edge node[left] {$a$} (st)
        (st) edge[loop below] node {$a$} (st);

\end{tikzpicture}
    		\caption{$\mu_{\mathit{out}}^*$}
    		\label{fig:flow-out}
    	\end{subfigure}
      \vspace{-3mm}
    	\caption{Example of outer macrostate analysis.
      (\subref{fig:flow-aut}) $\autex$ ($\bullet$ denotes accepting transitions).
      The initial \trub $\mu_0$ in (\subref{fig:flow-init}) is refined to
      $\mu_{\mathit{out}}^*$ in (\subref{fig:flow-out}).
      }
      \end{minipage}
  }
  \label{fig:outerflow}
\end{wrapfigure}
}

\vspace{-3.0mm}
\subsection{Outer Macrostate Analysis}\label{sec:label}
\vspace{-0.0mm}

We start with the simpler analysis, which is the \emph{outer macrostate
analysis}, which only looks at sizes of macrostates.
Recall that the rank~$r$ of every super-tight run in $\algschewe(\aut)$ does not
change, i.e., a~super tight run stays in \waiting as long as needed so that when
it jumps to \tight, it takes the rank~$r$ and never needs to decrease~it.
We~can use this fact to decrease the maximum rank of a~macrostate~$S$
in~$\skeletof \aut$.
In particular, let us consider all cycles going through~$S$.
For each of the cycles~$c$, we can bound the maximum rank of a~super-tight run
going through~$c$ by $2m-1$ where~$m$ is the smallest number of non-accepting
states occurring in any macrostate on~$c$ (from the definition, the rank of
a~tight ranking does not depend on accepting states).
Then we can infer that the maximum rank of any super-tight run going
through~$S$ is bounded by the maximum rank of any of the cycles going
through~$S$ (since~$S$ can never assume a~higher rank in any super-tight run).
Moreover, the rank of each cycle can also be estimated in a~more precise way,
e.g.\ using our elevator analysis.

Since the number of cycles in~$\skeletof \aut$ can be large\footnote{%
$\skeletof \aut$~can be exponentially larger than~$\aut$ and the number of cycles
in~$\skeletof \aut$ can be exponential to the size of~$\skeletof \aut$, so the
total number of cycles can be double-exponential.
},
instead of their enumeration, we employ data flow analysis with the
value domain $\val = \omega$ (i.e, for every macrostate~$S$ of~$\skeletof \aut$,
we remember a~bound on the maximum rank of~$S$) and the following update~function:
\vspace{-1mm}
\begin{equation}
  \upd_{\mathit{out}}(\mu, S, R_1, \ldots, R_m) = \min\{ \mu(S), \max\{
    \mu(R_1), \ldots, \mu(R_m)  \} \} .
\end{equation}

\vspace{-1mm}
\noindent
Intuitively, the new bound on the maximum rank of~$S$ is taken as the smaller of
the previous bound~$\mu(S)$ and the largest of the bounds of all predecessors
of~$S$, and the new value is propagated forward by the data flow analysis.


\begin{wrapfigure}[10]{r}{6cm}
  \vspace{-0mm}
  \scalebox{0.96}{
      \begin{minipage}{6cm}
    	\begin{subfigure}{2.3cm}
    		\hspace*{-2mm}
    		\begin{tikzpicture}[->,>=stealth',shorten >=0pt,auto,node distance=1.5cm,
                    scale=0.6,transform shape,initial text={}]
  \tikzstyle{every state}=[inner sep=3pt,minimum size=5pt]
  \tikzstyle{empty}=[]
  \tikzstyle{initstate}=[fill=yellow!30]

  \node[state,initial,initstate] (p) {$p$};
  \node[state,right of=p] (q) {$q$};
  \node[state,right of=q] (r) {$r$};
  \node[state,below of=q] (s) {$s$};

  \path (p) edge[loop below]  node {$a$} (p)
        (p) edge  node {$a$} (q)
        (q) edge  node {$a$} node[below] {$\bullet$} (r)
        (q) edge  node {$a$} (s)
        (r) edge[bend right=60]  node[above] {$a$} (p)
        (s) edge  node[right,yshift=2mm,xshift=-1mm] {$a$} node[left,yshift=-2mm,xshift=1mm] {$\bullet$} (p);

\end{tikzpicture}
    		\caption{$\autex$}
    		\label{fig:flow-aut}
    	\end{subfigure}
    	~
    	\begin{subfigure}{1.6cm}
    		\centering
    		\begin{tikzpicture}[->,>=stealth',shorten >=0pt,auto,node distance=1.8cm,
                    scale=0.6,transform shape,initial text={}]
  \tikzstyle{every state}=[inner sep=3pt,minimum size=5pt,rectangle,rounded corners=1mm]
  \tikzstyle{empty}=[]
  \tikzstyle{initstate}=[fill=yellow!30]
  \tikzstyle{wobbly}=[decorate, decoration={snake,amplitude=.2mm,segment length=2mm,post length=1mm}]

  \node[state,initial,initstate,label=0:{${:}1$}] (r) {$\{p\}$};
  \node[state, below of=r,node distance=12mm,label=0:{${:}3$}] (rst) {$\{p,q\}$};
  \node[state, below of=rst,node distance=12mm,label=0:{${:}7$}] (st) {$\{p,q,r,s\}$};

  \path (r) edge node[left] {$a$} (rst)
        (rst) edge node[left] {$a$} (st)
        (st) edge[loop below] node {$a$} (st);

\end{tikzpicture}
    		\caption{$\mu_0$}
    		\label{fig:flow-init}
    	\end{subfigure}
      ~
    	\begin{subfigure}{1.6cm}
    		\centering
    		\begin{tikzpicture}[->,>=stealth',shorten >=0pt,auto,node distance=1.8cm,
                    scale=0.6,transform shape,initial text={}]
  \tikzstyle{every state}=[inner sep=3pt,minimum size=5pt,rectangle,rounded corners=1mm]
  \tikzstyle{empty}=[]
  \tikzstyle{initstate}=[fill=yellow!30]
  \tikzstyle{wobbly}=[decorate, decoration={snake,amplitude=.2mm,segment length=2mm,post length=1mm}]

  \node[state,initial,initstate,label=0:{${:}1$}] (r) {$\{p\}$};
  \node[state, below of=r,node distance=12mm,label=0:{${:}\textcolor{red}{1}$}] (rst) {$\{p,q\}$};
  \node[state, below of=rst,node distance=12mm,label=0:{${:}7$}] (st) {$\{p,q,r,s\}$};

  \path (r) edge node[left] {$a$} (rst)
        (rst) edge node[left] {$a$} (st)
        (st) edge[loop below] node {$a$} (st);

\end{tikzpicture}
    		\caption{$\mu_{\mathit{out}}^*$}
    		\label{fig:flow-out}
    	\end{subfigure}
      \vspace{-3mm}
    	\caption{Example of outer macrostate analysis.
      (\subref{fig:flow-aut}) $\autex$ ($\bullet$ denotes accepting transitions).
      The initial \trub $\mu_0$ in (\subref{fig:flow-init}) is refined to
      $\mu_{\mathit{out}}^*$ in (\subref{fig:flow-out}).
      }
      \end{minipage}
  }
  \label{fig:outerflow}
\end{wrapfigure}

\vspace{-1mm}
\beginexample\label{ex:outer}
\mbox{Consider the BA~$\autex$ in} \cref{fig:flow-aut}.
When started from the initial \trub
$\mu_0 = \{\{p\} \mapsto 1, \{p,q\} \mapsto 3$, $\{p,q,r,s\} \mapsto 7\}$
(\cref{fig:flow-init}), outer macro\-state analysis decreases the maximum
rank estimate for $\{p,q\}$ to~1, since $\min\{\mu_0(\{p,q\},
\max\{\mu_0(\{p\})\}\} = \min\{3, 1\} = 1$.
The estimate for $\{p,q,r,s\}$ is not affected, because $\min\{7, \max\{1,7\}\}
= 7$ (\cref{fig:flow-out}).
\qed
%

%
\vspace{2mm}
\begin{restatable}{lemma}{theOutTrub}\label{lem:out-trub}
	If $\mu$ is a~\trub, then $\variant \mu \{S \mapsto \upd_{\mathit{out}}(\mu, S, R_1, \ldots, R_m)\}$ is
  a~\trub.
\end{restatable}
%
%
\vspace{-3mm}
\begin{restatable}{corollary}{theOutTerm}
	When started with a~\trub $\mu_0$, the outer macrostate analysis terminates
  and returns a~\trub~$\mu_{\mathit{out}}^*$.
\end{restatable}


\newcommand{\innermacroanal}[0]{
\begin{wrapfigure}[9]{r}{7.3cm}
\vspace*{-0mm}
\hspace*{0mm}
\begin{minipage}{7.3cm}
\SetKwProg{Fn}{}{:}{}
\setlength{\interspacetitleruled}{0pt}%
\setlength{\algotitleheightrule}{0pt}%
\begin{function}[H]
\Fn{$\upd_{\mathit{in}}(\mu, S, R_1, \ldots, R_m)$}{
	\ForEach{$1 \leq i \leq m$ \textbf{and} $a\in\Sigma$} {
    \If{$\trans(R_i,a) = S$}{
      $g_i^a \gets \maxrsaof{R_i}{\mu(R_i)}$
    }
  }
	$\theta \gets \mu(S) \sqcap \bigsqcup \{g_i^a \mid g_i^a \text{ is
  defined}\}$\;
	\lIf{$\rankof{\theta}$ is even}{$\theta \gets \decof{\theta}$\label{line:rankeven}}
	\Return $\theta$\;
}
  \vspace*{1mm}
	\label{alg:fin-update}
\end{function}
\end{minipage}
\end{wrapfigure}
}

\newcommand{\figinnerflow}[0]{
\begin{wrapfigure}[11]{r}{3.3cm}
\vspace*{-8mm}
\hspace*{-2mm}
\begin{minipage}{3.5cm}
  \begin{tikzpicture}[->,>=stealth',shorten >=0pt,auto,node distance=8mm,
                    scale=0.65,transform shape,initial text={}]
  \tikzstyle{every state}=[inner sep=3pt,minimum size=5pt]
  \tikzstyle{empty}=[]
  \tikzstyle{initstate}=[fill=yellow!30]
  \tikzstyle{wobbly}=[decorate, decoration={snake,amplitude=.2mm,segment length=2mm,post length=1mm}]

  \node (fp0) {$\{p{:}1,$};
  \node[right of=fp0] (fq0) {$q{:}1\}$};

  \node[right of=fq0,node distance=12mm] (fp1) {$\{p{:}7,$};
  \node[right of=fp1] (fq1) {$q{:}7,$};
  \node[right of=fq1] (fr1) {$r{:}7,$};
  \node[right of=fr1] (fs1) {$s{:}7 \}$};

  \node[below of=fp1, node distance=20mm,xshift=-9mm] (fp2) {$\{p{:}6,$};
  \node[right of=fp2] (fq2) {$q{:}7,$};
  \node[right of=fq2] (fr2) {$r{:}7,$};
  \node[right of=fr2] (fs2) {$s{:}7 \}$};

  \path[dashed] (fp0) edge (fp2)
      (fp0) edge (fq2)
      (fq0) edge (fr2)
      (fq0) edge (fs2);

  \path (fp1) edge (fp2)
        (fp1) edge (fq2)
        (fq1) edge (fr2)
        (fq1) edge (fs2)
        (fr1) edge (fp2)
        (fs1) edge[red] (fp2);

  \node[below of=fp2, node distance=16mm] (fp3) {$\{p{:}6,$};
  \node[right of=fp3] (fq3) {$q{:}6,$};
  \node[right of=fq3] (fr3) {$r{:}7,$};
  \node[right of=fr3] (fs3) {$s{:}7 \}$};

  \path (fp2) edge (fp3)
        (fq2) edge (fr3)
        (fq2) edge (fs3)
        (fr2) edge (fp3)
        (fs2) edge (fp3)
        (fp2) edge[red] (fq3);

  \node[below of=fp3, node distance=16mm] (fp4) {$\{p{:}6,$};
  \node[right of=fp4] (fq4) {$q{:}6,$};
  \node[right of=fq4] (fr4) {$r{:}6,$};
  \node[right of=fr4] (fs4) {$s{:}6 \}$};

  \path (fp3) edge (fp4)
        (fp3) edge (fq4)
        (fr3) edge (fp4)
        (fs3) edge (fp4)
        (fq3) edge[red] (fr4)
        (fq3) edge[red] (fs4);

  \node[below of=fp4, node distance=10mm] (fp5) {$\{p{:}5,$};
  \node[right of=fp5] (fq5) {$q{:}5,$};
  \node[right of=fq5] (fr5) {$r{:}5,$};
  \node[right of=fr5] (fs5) {$s{:}5 \}$};

  \draw[wobbly,red] ([xshift=-5mm]fr4.south) -- node[right] {$\decf$} ([xshift=-5mm]fr5.north);

\end{tikzpicture}
\end{minipage}
\end{wrapfigure}
}

\vspace{-4.0mm}
\subsection{Inner Macrostate Analysis}\label{sec:label}
\vspace{-0.0mm}
%
Our second analysis, called \emph{inner macrostate analysis}, looks deeper into
super-tight runs in $\algschewe(\aut)$.
In particular, compared with the outer macrostate analysis from the previous
section---which only looks at the \emph{ranks}, i.e., the bounds on the numbers
in the rankings---, inner macrostate analysis looks at how the \emph{rankings}
assign concrete values to the \emph{states} of~$\aut$ \emph{inside the macrostates}.

Inner macrostate analysis is based on the following.
Let~$\rho$ be a~super-tight run of $\algschewe(\aut)$ on~$\word \notin \langof
\aut$ and $\sofi$ be a~macrostate from \tight.
Because~$\rho$ is super-tight, we know that the rank~$f(q)$ of a~state $q \in S$
is bounded by the ranks of the predecessors of~$q$.
This holds because in super-tight runs, the ranks are only \emph{as high as
necessary}; if the rank of~$q$ were higher than the ranks of its predecessors,
this would mean that we may wait in \waiting longer and only jump to~$q$ with
a~lower rank later.

Let us introduce some necessary notation.
Let $f, f' \in \R$ be rankings (i.e., $f,f'\colon Q \to \omega$).
We use $f \sqcup f'$ to denote the ranking $\{ q \mapsto \max\{f(q),
f'(q)\} \mid q\in Q \}$, and $f \sqcap f'$ to denote the ranking $\{ q \mapsto
\min\{ f(q), f'(q)\} \mid q\in Q \}$.
Moreover, we define
$\maxrsa(f) = \max_\leq\{ f' \in \R \mid f \transconsist_S^a f' \}$ and a~function
$\decf\colon \R \to \R$ such that $\decof \theta$ is the ranking~$\theta'$ for
which
\begin{equation}
\theta'(q) = \begin{cases}
  \theta(q) \monus 1 & \text{if } \theta(q) = \rankof \theta \text{ and } q \notin \accstates, \\
  \evenceil{\theta(q) \monus 1} & \text{if } \theta(q) = \rankof \theta \text{ and } q \in \accstates, \\
  \theta(q) & \text{otherwise.}
\end{cases}
\end{equation}
%
Intuitively, $\maxrsa(f)$ is the (pointwise) maximum ranking that can be reached
from macrostate~$S$ with ranking~$f$ over~$a$ (it is easy to see that there is
a~unique such maximum ranking) and $\decof \theta$ decreases the maximum ranks
in a~ranking~$\theta$ by one (or~by two for even maximum ranks and accepting states).

The analysis uses the value domain $\val = \R$ (i.e., each macrostate of
$\skeletof \aut$ is assigned a~ranking giving an upper bound on the rank of each
state in the macrostate) and the update function~$\upd_{\mathit{in}}$ given in
the right-hand side of the page.
Intuitively, $\upd_{\mathit{in}}$
\innermacroanal
updates~$\mu(q)$ for every $q \in S$ to hold
the maximum rank compatible with the ranks of its predecessors.
We note line \cref{line:rankeven}, which makes use of the fact that we can only
consider tight rankings (whose rank is odd), so we can decrease the estimate
using the function $\decf$ defined above.

\vspace{-1mm}

\begin{wrapfigure}[11]{r}{3.3cm}
\vspace*{-8mm}
\hspace*{-2mm}
\begin{minipage}{3.5cm}
  \begin{tikzpicture}[->,>=stealth',shorten >=0pt,auto,node distance=8mm,
                    scale=0.65,transform shape,initial text={}]
  \tikzstyle{every state}=[inner sep=3pt,minimum size=5pt]
  \tikzstyle{empty}=[]
  \tikzstyle{initstate}=[fill=yellow!30]
  \tikzstyle{wobbly}=[decorate, decoration={snake,amplitude=.2mm,segment length=2mm,post length=1mm}]

  \node (fp0) {$\{p{:}1,$};
  \node[right of=fp0] (fq0) {$q{:}1\}$};

  \node[right of=fq0,node distance=12mm] (fp1) {$\{p{:}7,$};
  \node[right of=fp1] (fq1) {$q{:}7,$};
  \node[right of=fq1] (fr1) {$r{:}7,$};
  \node[right of=fr1] (fs1) {$s{:}7 \}$};

  \node[below of=fp1, node distance=20mm,xshift=-9mm] (fp2) {$\{p{:}6,$};
  \node[right of=fp2] (fq2) {$q{:}7,$};
  \node[right of=fq2] (fr2) {$r{:}7,$};
  \node[right of=fr2] (fs2) {$s{:}7 \}$};

  \path[dashed] (fp0) edge (fp2)
      (fp0) edge (fq2)
      (fq0) edge (fr2)
      (fq0) edge (fs2);

  \path (fp1) edge (fp2)
        (fp1) edge (fq2)
        (fq1) edge (fr2)
        (fq1) edge (fs2)
        (fr1) edge (fp2)
        (fs1) edge[red] (fp2);

  \node[below of=fp2, node distance=16mm] (fp3) {$\{p{:}6,$};
  \node[right of=fp3] (fq3) {$q{:}6,$};
  \node[right of=fq3] (fr3) {$r{:}7,$};
  \node[right of=fr3] (fs3) {$s{:}7 \}$};

  \path (fp2) edge (fp3)
        (fq2) edge (fr3)
        (fq2) edge (fs3)
        (fr2) edge (fp3)
        (fs2) edge (fp3)
        (fp2) edge[red] (fq3);

  \node[below of=fp3, node distance=16mm] (fp4) {$\{p{:}6,$};
  \node[right of=fp4] (fq4) {$q{:}6,$};
  \node[right of=fq4] (fr4) {$r{:}6,$};
  \node[right of=fr4] (fs4) {$s{:}6 \}$};

  \path (fp3) edge (fp4)
        (fp3) edge (fq4)
        (fr3) edge (fp4)
        (fs3) edge (fp4)
        (fq3) edge[red] (fr4)
        (fq3) edge[red] (fs4);

  \node[below of=fp4, node distance=10mm] (fp5) {$\{p{:}5,$};
  \node[right of=fp5] (fq5) {$q{:}5,$};
  \node[right of=fq5] (fr5) {$r{:}5,$};
  \node[right of=fr5] (fs5) {$s{:}5 \}$};

  \draw[wobbly,red] ([xshift=-5mm]fr4.south) -- node[right] {$\decf$} ([xshift=-5mm]fr5.north);

\end{tikzpicture}
\end{minipage}
\end{wrapfigure}

\beginexample
Let us continue in \cref{ex:outer} and perform inner
macrostate analysis starting with the \trub
$\{\{p{:}1\}, \{p{:}1, q{:}1\}, \{p{:}7,q{:}7,r{:}7,s{:}7\}\}$ obtained
from~$\mu_{\mathit{out}}^*$.
We show three iterations of the algorithm for $\{p,q,r,s\}$ in the right-hand
side
(we do not show $\{p,q\}$ except the first iteration
since it does not affect intermediate steps).
We can notice that in the three iterations, we could decrease the maximum rank
estimate to $\{p{:}6,q{:}6,r{:}6,s{:}6\}$ due to the accepting transitions
from~$r$ and~$s$.
In the last of the three iterations, when all states have the even rank~$6$, the
condition on \cref{line:rankeven} would become true and the rank of all states
would be decremented to~$5$ using~$\decf$.
Then, again, the accepting transitions from~$r$ and~$s$ would decrease the rank
of~$p$ to~$4$, which would be propagated to~$q$ and so on.
Eventually, we would arrive to the \trub $\{p{:}1,q{:}1,r{:}1,s{:}1\}$, which
could not be decreased any more, since $\{p{:}1,q{:}1\}$ forces the ranks of~$r$
and~$s$ to stay at~1.
\qed

%
\begin{restatable}{lemma}{theInnerTrub}\label{lem:inner-trub}
	If $\mu$ is a~\trub, then $\variant{\mu}{\{S \mapsto \upd_{\mathit{in}}(\mu,
  S, R_1, \ldots, R_m)\}}$ is
  a~\trub.
\end{restatable}
%
%
\vspace{-4mm}
\begin{restatable}{corollary}{theInTerm}
	When started with a~\trub $\mu_0$, the inner macrostate analysis terminates
  and returns a~\trub~$\mu_{\mathit{in}}^*$.
\end{restatable}

\vspace{-7.0mm}
\section{Experimental Evaluation}\label{sec:experiments}
\vspace{-2.0mm}

\paragraph{Used tools and evaluation environment.}
We implemented the techniques described in the previous sections as an extension
of the tool \ranker~\cite{ranker} (written in C++).
Speaking in the terms of~\cite{HavlenaL2021},
the heuristics were implemented on top of the $\rankermaxrank$ configuration (we
refer to this previous version as $\rankerold$).
We tested the correctness of our implementation using \spot's
\texttt{autcross} on all BAs in our benchmark.
We compared modified \ranker with other state-of-the-art tools, namely,
\goal~\cite{goal} (implementing \piterman~\cite{piterman2006nondeterministic},
\algschewe~\cite{Schewe09}, \safra~\cite{safra1988complexity},
and \fribourg~\cite{fribourg}),
\spot~2.9.3~\cite{spot} (implementing Redziejowski's
algorithm~\cite{Redziejowski12}), \seminator~\cite{seminator},
\ltldstar~0.5.4~\cite{KleinB07}, and \roll~\cite{roll}.
All tools were set to the mode where they output an automaton with the standard
state-based B\"{u}chi acceptance condition.
The experimental evaluation was performed on a~64-bit \textsc{GNU/Linux Debian}
workstation with an Intel(R) Xeon(R) CPU E5-2620 running at 2.40\,GHz with
32\,GiB of RAM and using a~timeout of 5\,minutes.

\vspace*{-3mm}
\paragraph{Datasets.}
As the source of our benchmark, we use the two following datasets:
\begin{inparaenum}[(i)]
  \item  \dsrandom containing 11,000 BAs over a~two letter alphabet used in~\cite{tsai-compl}, which
    were randomly generated via the Tabakov-Vardi approach~\cite{TabakovV05},
    starting from 15 states and with various different parameters;
  \item  \dsltl with 1,721 BAs over larger alphabets (up to 128 symbols) used in~\cite{seminator}, which were
    obtained from LTL formulae from literature (221) or randomly generated (1,500).
\end{inparaenum}
We preprocessed the automata using
\rabit~\cite{MayrC13} and \spot's \autfilt (using the \verb=--high=
simplification level), transformed them to state-based acceptance BAs (if they
were not already), and converted to the HOA format~\cite{BabiakBDKKM0S15}.
From this set, we removed automata that were
\begin{inparaenum}[(i)]
  \item  semi-deterministic,
  \item  inherently weak,
  \item  unambiguous, or
  \item  have an empty language,
\end{inparaenum}
since for these automata types there exist more efficient complementation
procedures than for unrestricted
BAs~\cite{BlahoudekHSST16,seminator,BoigelotJW01,li-unambigous}.
In~the end, we were left with \textbf{2,592} (\dsrandom) and \textbf{414}
(\dsltl) \emph{hard} automata.
We use \dsall to denote their union (\textbf{3,006} BAs). Of these hard automata, 458 were
elevator automata.

\newcommand{\figrankbased}[0]{
\begin{figure}[t]
  \begin{subfigure}[b]{0.49\linewidth}
  \begin{center}
  \includegraphics[width=\linewidth,keepaspectratio]{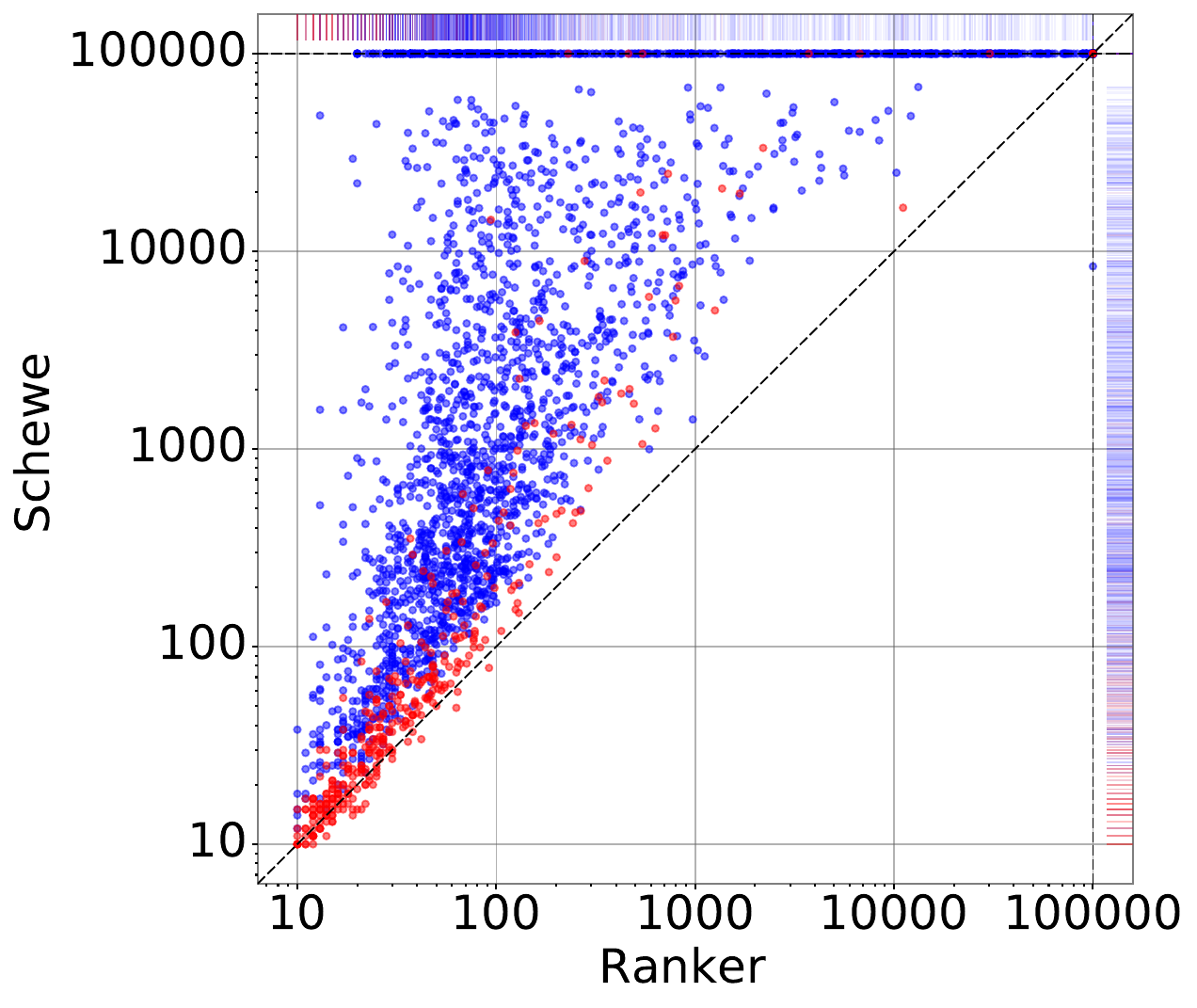}
  \end{center}
  \vspace{-5mm}
  \caption{$\ranker$ vs $\algschewe$}
  \label{fig:schewe}
  \end{subfigure}
  \begin{subfigure}[b]{0.49\linewidth}
  \begin{center}
  \includegraphics[width=\linewidth,keepaspectratio]{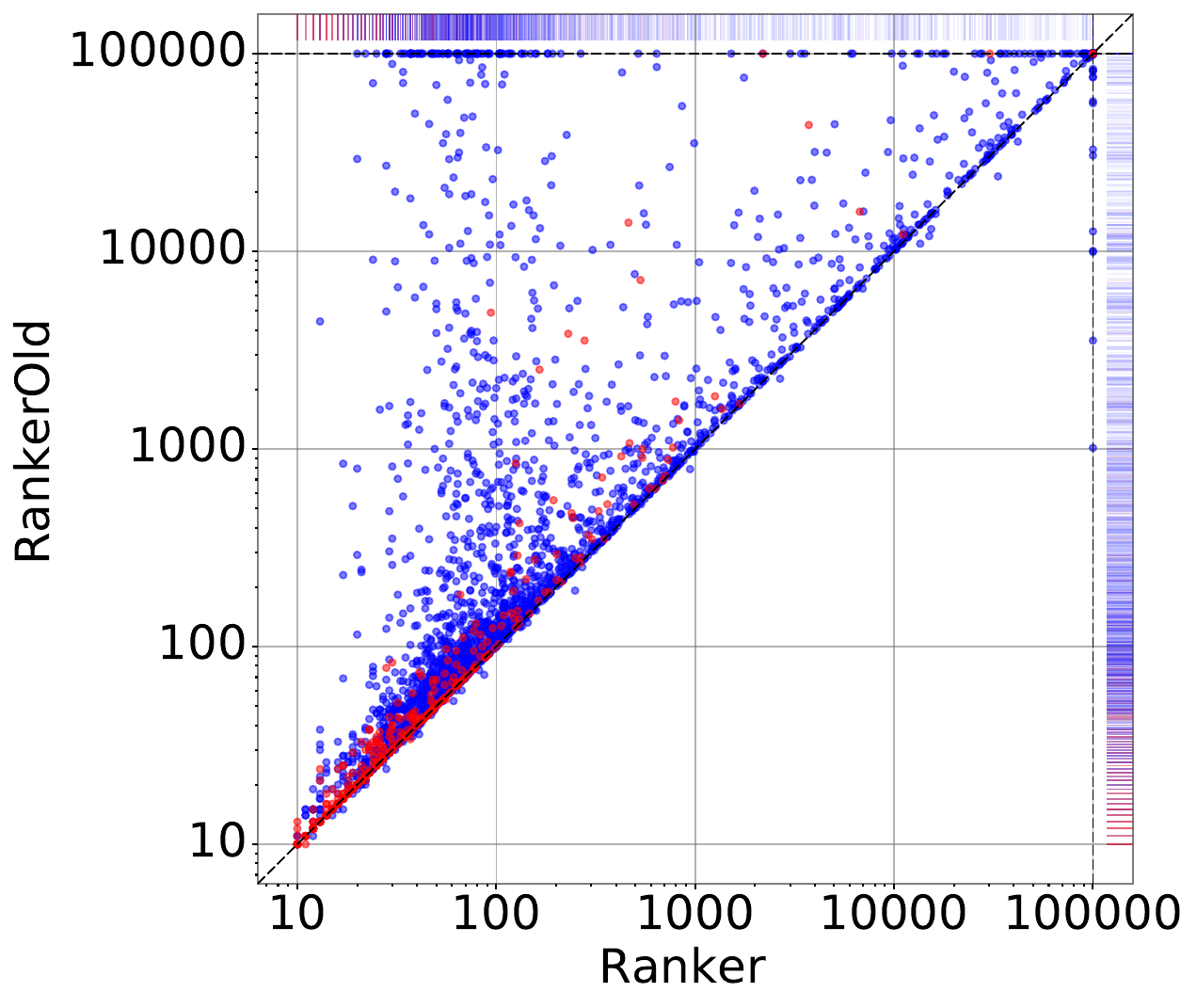}
  \end{center}
  \vspace{-5mm}
  \caption{$\ranker$ vs $\rankerold$}
  \label{fig:rankerold}
  \end{subfigure}
  \vspace{-2mm}
\caption{Comparison of the state space generated by our optimizations and other
  rank-based procedures (horizontal and vertical dashed lines represent timeouts).
  Blue data points are from \dsrandom and red data points are from \dsltl.
  Axes are logarithmic.}
  \vspace*{-6mm}
\label{fig:rank-based}
\end{figure}
}

\newcommand{\figothers}[0]{
\begin{figure}[t]
  \vspace*{-2mm}
  \begin{subfigure}[b]{0.49\linewidth}
  \begin{center}
  \includegraphics[width=\linewidth,keepaspectratio]{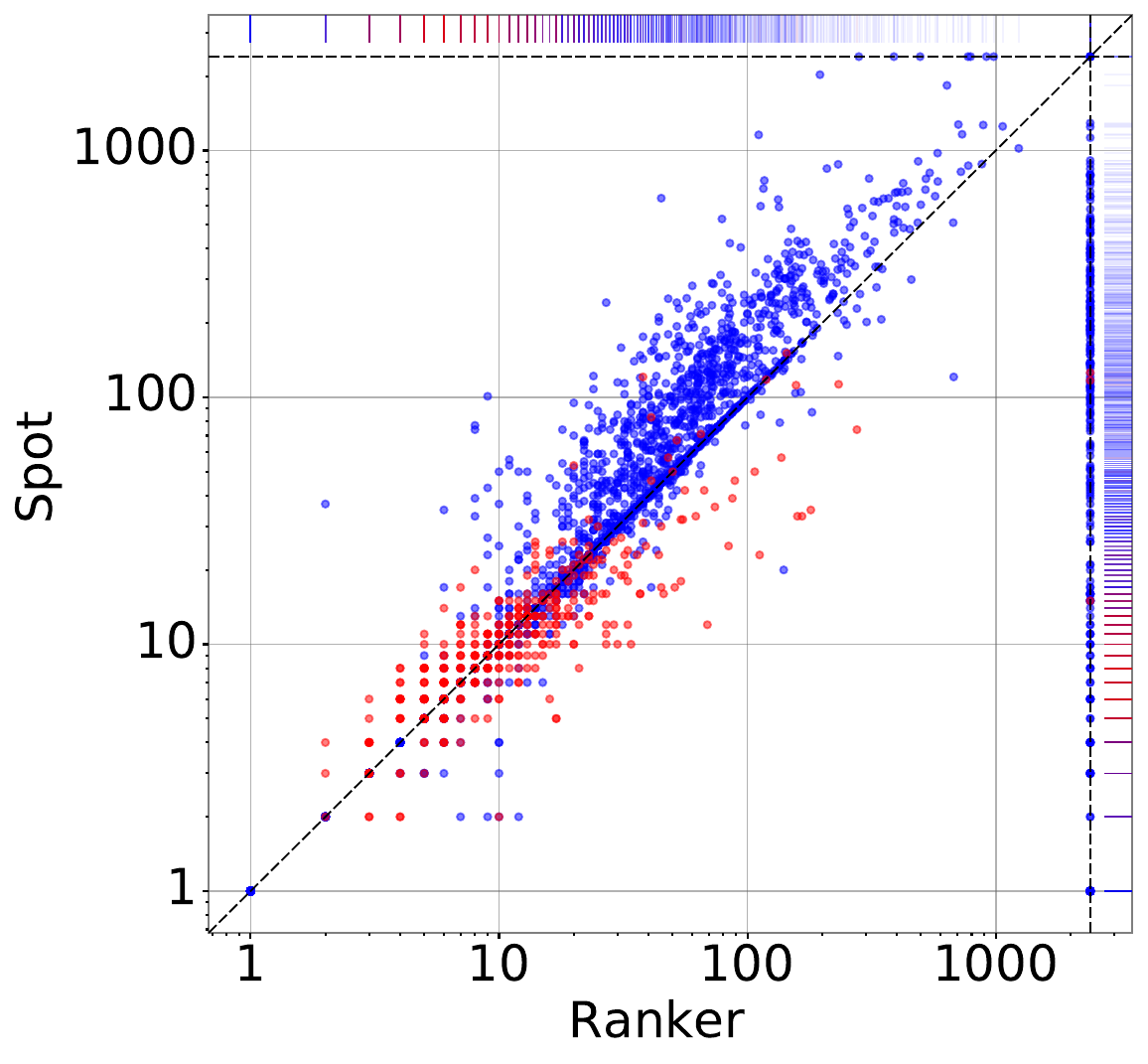}
  \end{center}
  \vspace{-5mm}
  \caption{$\ranker$ vs $\spot$}
  \label{fig:rankerspot}
  \end{subfigure}
  \begin{subfigure}[b]{0.49\linewidth}
  \begin{center}
  \includegraphics[width=\linewidth,keepaspectratio]{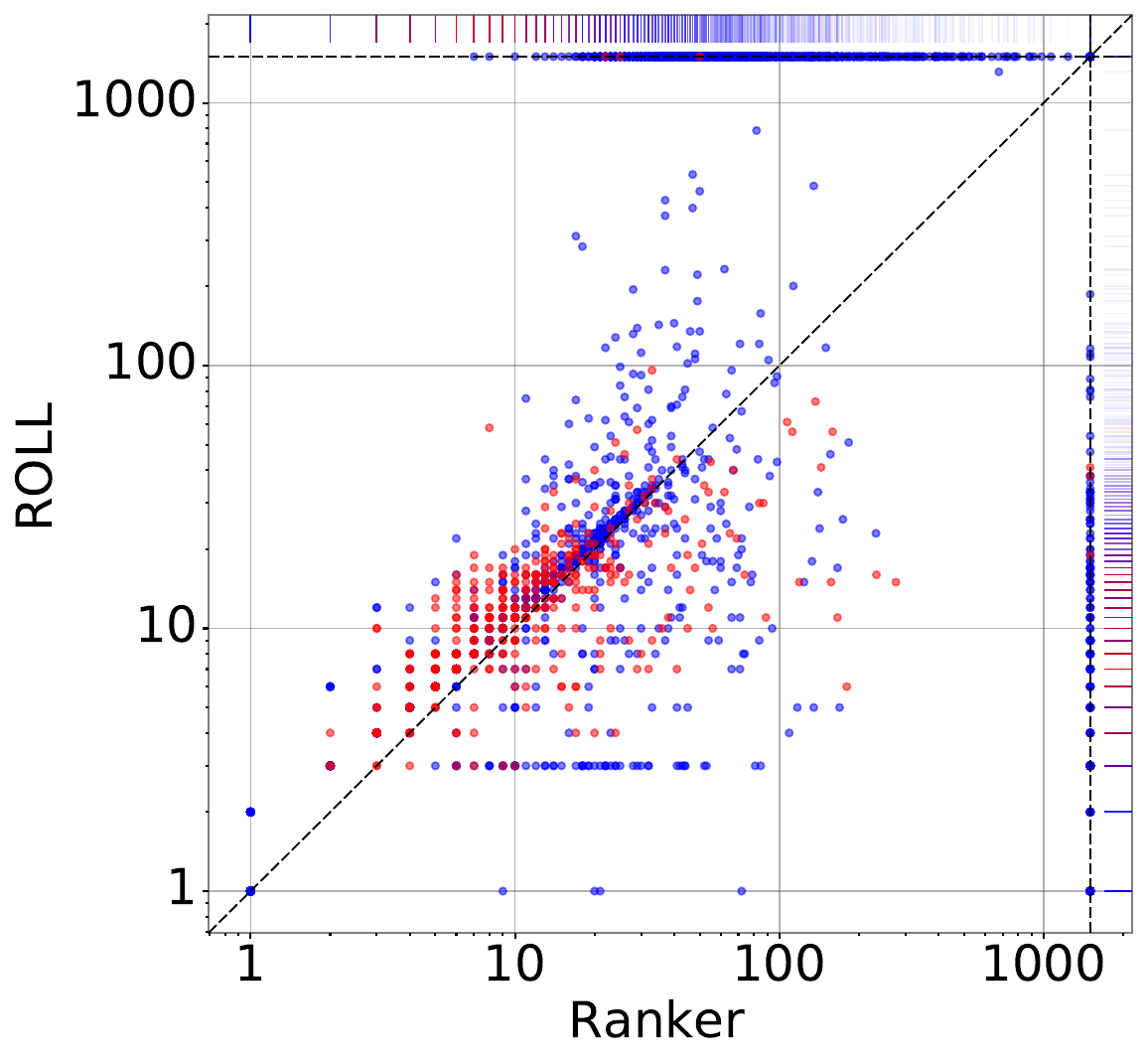}
  \end{center}
  \vspace{-5mm}
  \caption{$\ranker$ vs $\roll$}
  \label{fig:rankerroll}
  \end{subfigure}
  \vspace{-2mm}
\caption{Comparison of the complement size obtained by \ranker and other
  state-of-the-art tools (horizontal and vertical dashed lines represent
  timeouts).
  Axes are logarithmic.}
  \vspace*{-6mm}
\label{fig:comp-others}
\end{figure}
}

\figrankbased

\newcommand{\tabresults}[0]{
  \begin{table}[t]
  \vspace*{-1mm}
  \caption{Statistics for our experiments.
    The upper part compares various optimizations of the rank-based procedure
    (no postprocessing).
    The lower part compares \ranker to other approaches (with postprocessing).
    The left-hand side compares sizes of complement BAs and the right-hand side
    runtimes of the tools.
    The \textbf{wins} and \textbf{losses} columns give the number of times when \ranker
    was strictly better and worse.
    The values are given for the three datasets as ``\dsall (\dsrandom{} : \dsltl)''.
    Approaches in \goal are labelled with~\goalmark.
    }
  \label{tab:results}
  \hspace*{-2mm}
  \resizebox{1.02\linewidth}{!}{%
    \newcolumntype{g}{>{\columncolor{Gray!30}}r}
\newcolumntype{f}{>{\columncolor{Gray!30}}l}
\newcolumntype{h}{>{\columncolor{Gray!30}}c}
\begin{tabular}{lgghfrrclgghfrrcl|gghfrrclgghf}
\toprule
\multicolumn{1}{c}{\textbf{method}} & \multicolumn{4}{h}{\textbf{mean}}   & \multicolumn{4}{c}{\textbf{median}}   & \multicolumn{4}{h}{\textbf{wins}}   &   \multicolumn{4}{c}{\textbf{losses}} & \multicolumn{4}{h}{\textbf{mean runtime [s]}}   & \multicolumn{4}{c}{\textbf{median runtime [s]}}   & \multicolumn{4}{h}{\textbf{timeouts}}   \\
\midrule
 \rowcolor{GreenYellow}\ranker & 3812 & (4452 & \!:\! & 207) &  79 & (93  & \!:\! & 26) &      &       &       &      &     &      &       &      &  7.83 & (8.99  & \!:\! & 1.30) &                                  0.51 & (0.84                               & \!:\!                                 & 0.04)                               &                                   279 & (276                                        & \!:\!                                         & 3)                                      \\
 $\rankerold$                  & 7398 & (8688 & \!:\! & 358) & 141 & (197 & \!:\! & 29) & 2190 & (2011 & \!:\! & 179) & 111 & (107 & \!:\! & 4)   &  9.37 & (10.73 & \!:\! & 1.99) &                                  0.61 & (1.04                               & \!:\!                                 & 0.04)                               &                                   365 & (360                                        & \!:\!                                         & 5)                                      \\
 \algschewe~\goalmark          & 4550 & (5495 & \!:\! & 665) & 439 & (774 & \!:\! & 35) & 2640 & (2315 & \!:\! & 325) & 55  & (1   & \!:\! & 54)  & 21.05 & (24.28 & \!:\! & 7.80) &                                  6.57 & (7.39                               & \!:\!                                 & 5.21)                               &                                   937 & (928                                        & \!:\!                                         & 9)                                      \\
\midrule
 \rowcolor{GreenYellow}\ranker &   47 & (52   & \!:\! & 18)  &  22 & (27  & \!:\! & 10) &      &       &       &      &     &      &       &      &  7.83 & (8.99  & \!:\! & 1.30) &                                  0.51 & (0.84                               & \!:\!                                 & 0.04)                               &                                   279 & (276                                        & \!:\!                                         & 3)                                      \\
 \piterman~\goalmark           &   73 & (82   & \!:\! & 22)  &  28 & (34  & \!:\! & 14) & 1435 & (1124 & \!:\! & 311) & 416 & (360 & \!:\! & 56)  &  7.29 & (7.39  & \!:\! & 6.65) &                                  5.99 & (6.04                               & \!:\!                                 & 5.62)                               &                                    14 & (12                                         & \!:\!                                         & 2)                                      \\
 \safra~\goalmark              &   83 & (91   & \!:\! & 30)  &  29 & (35  & \!:\! & 17) & 1562 & (1211 & \!:\! & 351) & 387 & (350 & \!:\! & 37)  & 14.11 & (15.05 & \!:\! & 8.37) &                                  6.71 & (6.92                               & \!:\!                                 & 5.79)                               &                                   172 & (158                                        & \!:\!                                         & 14)                                     \\
 \spot                         &   75 & (85   & \!:\! & 15)  &  24 & (32  & \!:\! & 10) & 1087 & (936  & \!:\! & 151) & 683 & (501 & \!:\! & 182) &  0.86 & (0.99  & \!:\! & 0.06) &                                  0.02 & (0.02                               & \!:\!                                 & 0.02)                               &                                    13 & (13                                         & \!:\!                                         & 0)                                      \\
 \fribourg~\goalmark           &   91 & (104  & \!:\! & 13)  &  23 & (31  & \!:\! & 9)  & 1120 & (1055 & \!:\! & 65)  & 601 & (376 & \!:\! & 225) & 17.79 & (19.53 & \!:\! & 7.22) &                                  9.25 & (10.15                              & \!:\!                                 & 5.48)                               &                                    81 & (80                                         & \!:\!                                         & 1)                                      \\
 \ltldstar                     &   73 & (82   & \!:\! & 21)  &  28 & (34  & \!:\! & 13) & 1465 & (1195 & \!:\! & 270) & 465 & (383 & \!:\! & 82)  &  3.31 & (3.84  & \!:\! & 0.11) &                                  0.04 & (0.05                               & \!:\!                                 & 0.02)                               &                                   136 & (130                                        & \!:\!                                         & 6)                                      \\
 \seminator                    &   79 & (91   & \!:\! & 15)  &  21 & (29  & \!:\! & 10) & 1266 & (1131 & \!:\! & 135) & 571 & (367 & \!:\! & 204) &  9.51 & (11.25 & \!:\! & 0.08) &                                  0.22 & (0.39                               & \!:\!                                 & 0.02)                               &                                   363 & (362                                        & \!:\!                                         & 1)                                      \\
 \roll                         &   18 & (19   & \!:\! & 14)  &  10 & (9   & \!:\! & 11) & 2116 & (1858 & \!:\! & 258) & 569 & (443 & \!:\! & 126) & 31.23 & (37.85 & \!:\! & 7.28) &                                  8.19 & (12.23                              & \!:\!                                 & 2.74)                               &                                  1109 & (1106                                       & \!:\!                                         & 3)                                      \\
\bottomrule
\end{tabular}

  }
  \end{table}
}

\vspace{-3.0mm}
\subsection{Generated State Space}\label{sec:label}
\vspace{-2.0mm}

In our first experiment, we evaluated the effectiveness of our heuristics for
pruning the generated state space by comparing the sizes of complemented BAs
without postprocessing.
This use case is directed towards applications where postprocessing is
irrelevant, such as inclusion or equivalence checking of BAs.

We focused on a~comparison with two less optimized versions of the rank-based
complementation procedure:
$\algschewe$ (the version ``Reduced Average Outdegree'' from~\cite{Schewe09}
implemented in \goal under \texttt{-m rank -tr -ro}) and its optimization
$\rankerold$.
The scatter plots in \cref{fig:rank-based} compare the numbers of states of
automata generated by \ranker and the other algorithms and the upper part of
\cref{tab:results} gives summary statistics.
Observe that our optimizations from this paper drastically reduced the
generated search space compared with both \algschewe and $\rankerold$ (the mean
for \algschewe is lower than for $\rankerold$ due to its much higher number of
timeouts); from \cref{fig:rankerold} we can see that the improvement was in
many cases \emph{exponential} even when compared with our previous optimizations
in $\rankerold$.
The median (which is a~more meaningful indicator with the presence of timeouts)
decreased by 44\,\% w.r.t.\ $\rankerold$, and we also reduced the number of
timeouts by 23\,\%.
Notice that the numbers for the \dsltl dataset do not differ as much as for
\dsrandom, witnessing the easier structure of the BAs in \dsltl.

  \begin{table}[t]
  \vspace*{-1mm}
  \caption{Statistics for our experiments.
    The upper part compares various optimizations of the rank-based procedure
    (no postprocessing).
    The lower part compares \ranker to other approaches (with postprocessing).
    The left-hand side compares sizes of complement BAs and the right-hand side
    runtimes of the tools.
    The \textbf{wins} and \textbf{losses} columns give the number of times when \ranker
    was strictly better and worse.
    The values are given for the three datasets as ``\dsall (\dsrandom{} : \dsltl)''.
    Approaches in \goal are labelled with~\goalmark.
    }
  \label{tab:results}
  \hspace*{-2mm}
  \resizebox{1.02\linewidth}{!}{%
    \newcolumntype{g}{>{\columncolor{Gray!30}}r}
\newcolumntype{f}{>{\columncolor{Gray!30}}l}
\newcolumntype{h}{>{\columncolor{Gray!30}}c}
\begin{tabular}{lgghfrrclgghfrrcl|gghfrrclgghf}
\toprule
\multicolumn{1}{c}{\textbf{method}} & \multicolumn{4}{h}{\textbf{mean}}   & \multicolumn{4}{c}{\textbf{median}}   & \multicolumn{4}{h}{\textbf{wins}}   &   \multicolumn{4}{c}{\textbf{losses}} & \multicolumn{4}{h}{\textbf{mean runtime [s]}}   & \multicolumn{4}{c}{\textbf{median runtime [s]}}   & \multicolumn{4}{h}{\textbf{timeouts}}   \\
\midrule
 \rowcolor{GreenYellow}\ranker & 3812 & (4452 & \!:\! & 207) &  79 & (93  & \!:\! & 26) &      &       &       &      &     &      &       &      &  7.83 & (8.99  & \!:\! & 1.30) &                                  0.51 & (0.84                               & \!:\!                                 & 0.04)                               &                                   279 & (276                                        & \!:\!                                         & 3)                                      \\
 $\rankerold$                  & 7398 & (8688 & \!:\! & 358) & 141 & (197 & \!:\! & 29) & 2190 & (2011 & \!:\! & 179) & 111 & (107 & \!:\! & 4)   &  9.37 & (10.73 & \!:\! & 1.99) &                                  0.61 & (1.04                               & \!:\!                                 & 0.04)                               &                                   365 & (360                                        & \!:\!                                         & 5)                                      \\
 \algschewe~\goalmark          & 4550 & (5495 & \!:\! & 665) & 439 & (774 & \!:\! & 35) & 2640 & (2315 & \!:\! & 325) & 55  & (1   & \!:\! & 54)  & 21.05 & (24.28 & \!:\! & 7.80) &                                  6.57 & (7.39                               & \!:\!                                 & 5.21)                               &                                   937 & (928                                        & \!:\!                                         & 9)                                      \\
\midrule
 \rowcolor{GreenYellow}\ranker &   47 & (52   & \!:\! & 18)  &  22 & (27  & \!:\! & 10) &      &       &       &      &     &      &       &      &  7.83 & (8.99  & \!:\! & 1.30) &                                  0.51 & (0.84                               & \!:\!                                 & 0.04)                               &                                   279 & (276                                        & \!:\!                                         & 3)                                      \\
 \piterman~\goalmark           &   73 & (82   & \!:\! & 22)  &  28 & (34  & \!:\! & 14) & 1435 & (1124 & \!:\! & 311) & 416 & (360 & \!:\! & 56)  &  7.29 & (7.39  & \!:\! & 6.65) &                                  5.99 & (6.04                               & \!:\!                                 & 5.62)                               &                                    14 & (12                                         & \!:\!                                         & 2)                                      \\
 \safra~\goalmark              &   83 & (91   & \!:\! & 30)  &  29 & (35  & \!:\! & 17) & 1562 & (1211 & \!:\! & 351) & 387 & (350 & \!:\! & 37)  & 14.11 & (15.05 & \!:\! & 8.37) &                                  6.71 & (6.92                               & \!:\!                                 & 5.79)                               &                                   172 & (158                                        & \!:\!                                         & 14)                                     \\
 \spot                         &   75 & (85   & \!:\! & 15)  &  24 & (32  & \!:\! & 10) & 1087 & (936  & \!:\! & 151) & 683 & (501 & \!:\! & 182) &  0.86 & (0.99  & \!:\! & 0.06) &                                  0.02 & (0.02                               & \!:\!                                 & 0.02)                               &                                    13 & (13                                         & \!:\!                                         & 0)                                      \\
 \fribourg~\goalmark           &   91 & (104  & \!:\! & 13)  &  23 & (31  & \!:\! & 9)  & 1120 & (1055 & \!:\! & 65)  & 601 & (376 & \!:\! & 225) & 17.79 & (19.53 & \!:\! & 7.22) &                                  9.25 & (10.15                              & \!:\!                                 & 5.48)                               &                                    81 & (80                                         & \!:\!                                         & 1)                                      \\
 \ltldstar                     &   73 & (82   & \!:\! & 21)  &  28 & (34  & \!:\! & 13) & 1465 & (1195 & \!:\! & 270) & 465 & (383 & \!:\! & 82)  &  3.31 & (3.84  & \!:\! & 0.11) &                                  0.04 & (0.05                               & \!:\!                                 & 0.02)                               &                                   136 & (130                                        & \!:\!                                         & 6)                                      \\
 \seminator                    &   79 & (91   & \!:\! & 15)  &  21 & (29  & \!:\! & 10) & 1266 & (1131 & \!:\! & 135) & 571 & (367 & \!:\! & 204) &  9.51 & (11.25 & \!:\! & 0.08) &                                  0.22 & (0.39                               & \!:\!                                 & 0.02)                               &                                   363 & (362                                        & \!:\!                                         & 1)                                      \\
 \roll                         &   18 & (19   & \!:\! & 14)  &  10 & (9   & \!:\! & 11) & 2116 & (1858 & \!:\! & 258) & 569 & (443 & \!:\! & 126) & 31.23 & (37.85 & \!:\! & 7.28) &                                  8.19 & (12.23                              & \!:\!                                 & 2.74)                               &                                  1109 & (1106                                       & \!:\!                                         & 3)                                      \\
\bottomrule
\end{tabular}

  }
  \end{table}

\vspace{-3.0mm}
\subsection{Comparison with Other Complementation Techniques}\label{sec:label}
\vspace{-2.0mm}

In our second experiment, we compared the improved \ranker with other
state-of-the-art tools. We were comparing sizes of output BAs, therefore, we
postprocessed each output automaton with \autfilt (simplification level
\verb=--high=).
Scatter plots are given in \cref{fig:comp-others}, where we compare \ranker with
\spot (which had the best results on average from the other tools except \roll)
and \roll, and summary statistics are in the lower part of \cref{tab:results}.
Observe that \ranker has by far the lowest mean (except \roll)
and the third lowest median (after \seminator and \roll, but with less timeouts).
Moreover, comparing the numbers in columns \textbf{wins} and \textbf{losses}
we can see that \ranker gives strictly better results than other tools (\textbf{wins}) more
often than the other way round (\textbf{losses}).

\enlargethispage{2mm}

In \cref{fig:rankerspot} see that indeed in the majority of cases \ranker gives
a~smaller BA than \spot, especially for harder BAs (\spot, however,
behaves slightly better on the simpler BAs from \dsltl).
The results in \cref{fig:rankerroll} do not seem so clear.
\roll uses a~learning-based approach---more heavyweight and completely
orthogonal to any of the other tools---and can in some cases output a~tiny
automaton, but does not scale, as observed by the number of timeouts much higher
than any other tool.
It is, therefore, positively surprising that \ranker could in most of the cases still obtain
a~much smaller automaton than \roll.

Regarding runtimes, the prototype implementation in \ranker is comparable
to \seminator, but slower than \spot and \ltldstar (\spot is the fastest tool).
Implementations of other approaches clearly do not target speed.
We note that the number of timeouts of \ranker is still higher than of some other
tools (in particular \piterman, \spot, \fribourg); further state space
\mbox{reduction targeting this particular issue is our future work.}

\vspace{-4.0mm}
\section{Related Work}\label{sec:related}
\vspace{-3.0mm}

BA complementation remains in the interest of researchers since their first introduction by B\"{u}chi in~\cite{buchi1962decision}.
Together with a hunt for efficient complementation techniques, the effort has
been put into establishing the lower bound. First, Michel showed that the lower
bound is $n!$ (approx. $(0.36n)^n$)~\cite{michel1988complementation} and later
Yan refined the result to $(0.76n)^n$~\cite{yan}.

\vspace{-0mm}
The complementation approaches can be roughly divided into several branches.
%
\emph{Ramsey-based complementation}, the very first complementation construction,
where the language of an input automaton is decomposed into a~finite number of equivalence classes,
was proposed by B\"{u}chi and was further enhanced
in~\cite{breuers-improved-ramsey}.
\emph{Determinization-based complementation} was
presented by Safra in~\cite{safra1988complexity} and later improved by
Piterman in~\cite{piterman2006nondeterministic} and Redziejowski
in~\cite{Redziejowski12}. Various optimizations for determinization of BAs were further proposed in~\cite{pigorov-determinization}.
The main idea of this approach is to convert
an input BA into an equivalent deterministic automaton with different acceptance
condition that can be easily complemented (e.g. Rabin automaton).
The complemented automaton is then converted back into a~BA
(often for the price of some blow-up).
\emph{Slice-based complementation}
tracks the acceptance condition using a~reduced abstraction on a~run tree~\cite{vardi2008automata,kahler2008complementation}.
\emph{A~learning-based approach} was introduced in~\cite{li2018learning,roll}.
Allred and Ultes-Nitsche then presented a~novel optimal complementation algorithm in~\cite{fribourg}.
%
For some special types of BAs, e.g., deterministic~\cite{Kurshan87},
semi-deterministic~\cite{BlahoudekHSST16}, or unambiguous~\cite{li-unambigous}, there exist specific complementation algorithms.
\emph{Semi-determinization based complementation}
converts an input BA into a~semi-deterministic BA~\cite{CourcoubetisY88}, which is then complemented~\cite{seminator}.

\figothers

\emph{Rank-based complementation}, studied
in~\cite{KupfermanV01,GurumurthyKSV03,FriedgutKV06,Schewe09,KarmarkarC09},
extends the subset construction for determinization of finite automata
by storing additional information in each macrostate
to track the acceptance condition of all runs of the input automaton.
Optimizations of an alternative (sub-optimal) rank-based construction from~\cite{KupfermanV01} going through \emph{alternating B\"{u}chi automata} were presented in~\cite{GurumurthyKSV03}.
Furthermore, the work in~\cite{KarmarkarC09} introduces an optimization of
\algschewe, in some cases producing smaller automata
(this construction is not compatible with our optimizations).
As shown in~\cite{ChenHL19},
the rank-based construction can be optimized using simulation relations.
We identified several heuristics that help reducing the size of the complement
in~\cite{HavlenaL2021}, which are compatible with the heuristics in this paper.


\vspace{-3mm}
\paragraph{Acknowledgements.}
We thank anonymous reviewers for their useful remarks that helped us improve
the quality of the paper.
This work was supported by the Czech Science Foundation project 20-07487S
and the FIT BUT internal project FIT-S-20-6427.

\bibliographystyle{splncs}
\bibliography{literature}

\begin{thebibliography}{10}
\providecommand{\url}[1]{\texttt{#1}}
\providecommand{\urlprefix}{URL }
\providecommand{\doi}[1]{https://doi.org/#1}

\bibitem{fribourg}
Allred, J.D., Ultes-Nitsche, U.: A simple and optimal complementation algorithm
  for {B\"{u}chi} automata. In: Proceedings of the Thirty third Annual IEEE
  Symposium on Logic in Computer Science (LICS 2018). pp. 46--55. IEEE Computer
  Society Press (July 2018)

\bibitem{BabiakBDKKM0S15}
Babiak, T., Blahoudek, F., Duret{-}Lutz, A., Klein, J.,
  K\v{r}et{\'{\i}}nsk{\'{y}}, J., M{\"{u}}ller, D., Parker, D., Strej\v{c}ek,
  J.: The {Hanoi} omega-automata format. In: Computer Aided Verification - 27th
  International Conference, {CAV} 2015, San Francisco, CA, USA, July 18-24,
  2015, Proceedings, Part {I}. Lecture Notes in Computer Science, vol.~9206,
  pp. 479--486. Springer (2015). \doi{10.1007/978-3-319-21690-4\_31}

\bibitem{semiDetComplementation}
Blahoudek, F., Heizmann, M., Schewe, S., Strej{\v{c}}ek, J., Tsai, M.H.:
  Complementing semi-deterministic b{\"u}chi automata. In: Tools and Algorithms
  for the Construction and Analysis of Systems. pp. 770--787. Springer Berlin
  Heidelberg, Berlin, Heidelberg (2016)

\bibitem{seminator}
Blahoudek, F., Duret-Lutz, A., Strej\v{c}ek, J.: {S}eminator~2 can complement
  generalized {B\"u}chi automata via improved semi-determinization. In:
  Proceedings of the 32nd International Conference on Computer-Aided
  Verification (CAV'20). Lecture Notes in Computer Science, vol. 12225, pp.
  15--27. Springer (Jul 2020)

\bibitem{BlahoudekHSST16}
Blahoudek, F., Heizmann, M., Schewe, S., Strej\v{c}ek, J., Tsai, M.:
  Complementing semi-deterministic {B{\"{u}}chi} automata. In: Tools and
  Algorithms for the Construction and Analysis of Systems - 22nd International
  Conference, {TACAS} 2016, Held as Part of the European Joint Conferences on
  Theory and Practice of Software, {ETAPS} 2016, Eindhoven, The Netherlands,
  April 2-8, 2016, Proceedings. Lecture Notes in Computer Science, vol.~9636,
  pp. 770--787. Springer (2016). \doi{10.1007/978-3-662-49674-9\_49}

\bibitem{BoigelotJW01}
Boigelot, B., Jodogne, S., Wolper, P.: On the use of weak automata for deciding
  linear arithmetic with integer and real variables. In: Automated Reasoning,
  First International Joint Conference, {IJCAR} 2001, Siena, Italy, June 18-23,
  2001, Proceedings. Lecture Notes in Computer Science, vol.~2083, pp.
  611--625. Springer (2001). \doi{10.1007/3-540-45744-5\_50}

\bibitem{breuers-improved-ramsey}
Breuers, S., L{\"o}ding, C., Olschewski, J.: Improved {R}amsey-based
  {B{\"u}chi} complementation. In: Proc. of FOSSACS'12. pp. 150--164. Springer
  (2012)

\bibitem{buchi1962decision}
B{\"u}chi, J.R.: On a decision method in restricted second order arithmetic.
  In: Proc. of International Congress on Logic, Method, and Philosophy of
  Science 1960. Stanford Univ. Press, Stanford (1962)

\bibitem{ChenHL19}
Chen, Y., Havlena, V., Leng{\'{a}}l, O.: Simulations in rank-based
  {B{\"{u}}chi} automata complementation. In: Programming Languages and Systems
  - 17th Asian Symposium, {APLAS} 2019, Nusa Dua, Bali, Indonesia, December
  1-4, 2019, Proceedings. Lecture Notes in Computer Science, vol. 11893, pp.
  447--467. Springer (2019). \doi{10.1007/978-3-030-34175-6\_23}

\bibitem{ChenHLLTTZ18}
Chen, Y., Heizmann, M., Leng{\'{a}}l, O., Li, Y., Tsai, M., Turrini, A., Zhang,
  L.: Advanced automata-based algorithms for program termination checking. In:
  Proceedings of the 39th {ACM} {SIGPLAN} Conference on Programming Language
  Design and Implementation, {PLDI} 2018, Philadelphia, PA, USA, June 18-22,
  2018. pp. 135--150. {ACM} (2018). \doi{10.1145/3192366.3192405}

\bibitem{CourcoubetisY88}
Courcoubetis, C., Yannakakis, M.: Verifying temporal properties of finite-state
  probabilistic programs. In: 29th Annual Symposium on Foundations of Computer
  Science, White Plains, New York, USA, 24-26 October 1988. pp. 338--345.
  {IEEE} Computer Society (1988). \doi{10.1109/SFCS.1988.21950}

\bibitem{spot}
Duret-Lutz, A., Lewkowicz, A., Fauchille, A., Michaud, T., Renault, {\'E}., Xu,
  L.: Spot 2.0 --- a framework for {LTL} and $\omega$-automata manipulation.
  In: Automated Technology for Verification and Analysis. pp. 122--129.
  Springer International Publishing, Cham (2016)

\bibitem{fogarty2009buchi}
Fogarty, S., Vardi, M.Y.: {B{\"u}chi} complementation and size-change
  termination. In: Proc. of TACAS'09. pp. 16--30. Springer (2009)

\bibitem{FriedgutKV06}
Friedgut, E., Kupferman, O., Vardi, M.: {B\"{u}chi} complementation made
  tighter. International Journal of Foundations of Computer Science
  \textbf{17},  851--868 (2006)

\bibitem{GurumurthyKSV03}
Gurumurthy, S., Kupferman, O., Somenzi, F., Vardi, M.Y.: On complementing
  nondeterministic {B{\"{u}}chi} automata. In: Correct Hardware Design and
  Verification Methods, 12th {IFIP} {WG} 10.5 Advanced Research Working
  Conference, {CHARME} 2003, L'Aquila, Italy, October 21-24, 2003, Proceedings.
  LNCS, vol.~2860, pp. 96--110. Springer (2003).
  \doi{10.1007/978-3-540-39724-3\_10}

\bibitem{HLS-S1S}
Havlena, V., Leng\'{a}l, O., \v{S}mahl\'{i}kov\'{a}, B.: Deciding {S1S}: Down
  the rabbit hole and through the looking glass. In: Proceedings of NETYS'21.
  pp. 215--222. No. 12754 in LNCS, Springer Verlag (2021).
  \doi{10.1007/978-3-030-91014-3\_15}

\bibitem{ranker}
Havlena, V., Leng\'{a}l, O., \v{S}mahl\'{i}kov\'{a}, B.: \ranker (2021),
  \url{https://github.com/vhavlena/ranker}

\bibitem{HavlenaL2021}
Havlena, V., Lengál, O.: Reducing ({To}) the {Ranks}: {Efficient}
  {Rank}-{Based} {Büchi} {Automata} {Complementation}. In: Proc.\ of
  CONCUR'21. LIPIcs, vol.~203, pp. 2:1--2:19. Schloss Dagstuhl, Dagstuhl,
  Germany (2021). \doi{10.4230/LIPIcs.CONCUR.2021.2}, iSSN: 1868-8969

\bibitem{heizmann2014termination}
Heizmann, M., Hoenicke, J., Podelski, A.: Termination analysis by learning
  terminating programs. In: Proc. of CAV'14. pp. 797--813. Springer (2014)

\bibitem{kahler2008complementation}
K{\"a}hler, D., Wilke, T.: Complementation, disambiguation, and determinization
  of {B{\"u}chi} automata unified. In: Proc. of ICALP'08. pp. 724--735.
  Springer (2008)

\bibitem{KarmarkarC09}
Karmarkar, H., Chakraborty, S.: On minimal odd rankings for {B{\"{u}}chi}
  complementation. In: Proc.\ of ATVA'09. LNCS, vol.~5799, pp. 228--243.
  Springer (2009). \doi{10.1007/978-3-642-04761-9\_18}

\bibitem{KleinB07}
Klein, J., Baier, C.: On-the-fly stuttering in the construction of
  deterministic \emph{omega} -automata. In: Proc.\ of CIAA'07. LNCS, vol.~4783,
  pp. 51--61. Springer (2007). \doi{10.1007/978-3-540-76336-9\_7}

\bibitem{KupfermanV01}
Kupferman, O., Vardi, M.Y.: Weak alternating automata are not that weak. {ACM}
  Trans. Comput. Log.  \textbf{2}(3),  408--429 (2001).
  \doi{10.1145/377978.377993}

\bibitem{Kurshan87}
Kurshan, R.P.: Complementing deterministic {B{\"{u}}chi} automata in polynomial
  time. J. Comput. Syst. Sci.  \textbf{35}(1),  59--71 (1987).
  \doi{10.1016/0022-0000(87)90036-5}

\bibitem{roll}
Li, Y., Sun, X., Turrini, A., Chen, Y., Xu, J.: {ROLL} 1.0: $\omega$-regular
  language learning library. In: Proc.\ of TACAS'19. LNCS, vol. 11427, pp.
  365--371. Springer (2019). \doi{10.1007/978-3-030-17462-0\_23}

\bibitem{li2018learning}
Li, Y., Turrini, A., Zhang, L., Schewe, S.: Learning to complement {B{\"u}chi}
  automata. In: Proc. of VMCAI'18. pp. 313--335. Springer (2018)

\bibitem{li-unambigous}
Li, Y., Vardi, M.Y., Zhang, L.: On the power of unambiguity in {B\"uchi}
  complementation. In: Proc. of GandALF'20. EPTCS, vol.~326, pp. 182--198. Open
  Publishing Association (2020). \doi{10.4204/EPTCS.326.12}

\bibitem{pigorov-determinization}
L{\"o}ding, C., Pirogov, A.: New optimizations and heuristics for
  determinization of b{\"u}chi automata. In: Automated Technology for
  Verification and Analysis. pp. 317--333. Springer International Publishing,
  Cham (2019). \doi{10.1007/978-3-030-31784-3\_18}

\bibitem{MayrC13}
Mayr, R., Clemente, L.: Advanced automata minimization. In: Proc. of POPL'13.
  pp. 63--74 (2013)

\bibitem{michel1988complementation}
Michel, M.: Complementation is more difficult with automata on infinite words.
  CNET, Paris  \textbf{15} (1988)

\bibitem{NielsenNH99}
Nielson, F., Nielson, H.R., Hankin, C.: Principles of program analysis.
  Springer (1999). \doi{10.1007/978-3-662-03811-6}

\bibitem{pecan}
Oei, R., Ma, D., Schulz, C., Hieronymi, P.: Pecan: An automated theorem prover
  for automatic sequences using b{\"{u}}chi automata. CoRR
  \textbf{abs/2102.01727} (2021), \url{https://arxiv.org/abs/2102.01727}

\bibitem{piterman2006nondeterministic}
Piterman, N.: From nondeterministic {B{\"u}chi} and {Streett} automata to
  deterministic parity automata. In: Proc. of LICS'06. pp. 255--264. IEEE
  (2006)

\bibitem{Redziejowski12}
Redziejowski, R.R.: An improved construction of deterministic omega-automaton
  using derivatives. Fundam. Informaticae  \textbf{119}(3-4),  393--406 (2012).
  \doi{10.3233/FI-2012-744}

\bibitem{safra1988complexity}
Safra, S.: On the complexity of $\omega$-automata. In: Proc. of FOCS'88. pp.
  319--327. IEEE (1988)

\bibitem{Schewe09}
Schewe, S.: {B{\"{u}}chi} complementation made tight. In: Proc.\ of STACS'09.
  LIPIcs, vol.~3, pp. 661--672. Schloss Dagstuhl (2009).
  \doi{10.4230/LIPIcs.STACS.2009.1854}

\bibitem{sistla1987complementation}
Sistla, A.P., Vardi, M.Y., Wolper, P.: {The Complementation Problem for
  {B{\"u}chi} Automata with Applications to Temporal Logic}. Theoretical
  Computer Science  \textbf{49}(2-3),  217--237 (1987)

\bibitem{TabakovV05}
Tabakov, D., Vardi, M.Y.: Experimental evaluation of classical automata
  constructions. In: Proc. of LPAR'05. pp. 396--411. Springer (2005)

\bibitem{tsai-compl}
Tsai, M.H., Fogarty, S., Vardi, M.Y., Tsay, Y.K.: State of {B{\"u}chi}
  complementation. In: Implementation and Application of Automata. pp.
  261--271. Springer Berlin Heidelberg, Berlin, Heidelberg (2011)

\bibitem{goal}
Tsai, M.H., Tsay, Y.K., Hwang, Y.S.: {GOAL} for games, omega-automata, and
  logics. In: Computer Aided Verification. pp. 883--889. Springer Berlin
  Heidelberg, Berlin, Heidelberg (2013)

\bibitem{vardi2008automata}
Vardi, M.Y., Wilke, T.: Automata: From logics to algorithms. Logic and Automata
   \textbf{2},  629--736 (2008)

\bibitem{yan}
Yan, Q.: Lower bounds for complementation of $\omega$-automata via the full
  automata technique. In: Automata, Languages and Programming. pp. 589--600.
  Springer Berlin Heidelberg, Berlin, Heidelberg (2006)

\end{thebibliography}

\ifTR
\newpage
\appendix

\section{Semi-Determinization of Elevator Automata}\label{sec:elev-semidet}

The deelevation procedure from \cref{sec:eff-elev} also provides a~good starting point for transforming
an elevator automaton into an SDBA.
For an elevator automaton~$\aut$ with no nondeterministic IWA component,
the output of $\algdeelev(\aut)$ is already an SDBA.
If~$\aut$, on the other hand, contains some nondeterministic IWA component, so will
$\algdeelev(\aut)$.
In such a~case, we can determinize each of the (terminal) IWA components in
$\algdeelev(\aut)$, which will cause a blow-up of at most $2^{|Q|}$%
\footnote{Nondeterministic inherently weak BAs are stricly more
expressive than deterministic inherently weak BAs (observed, e.g., by the
$\omega$-regular language $(a+b)^*b^\omega$ for which no deterministic BA exists
but is recognized by a~nondeterministic weak BA).
It is, however, folklore knowledge that inherently weak components can
be determinized by the powerset construction.},
which would yield an SDBA with at most $2|Q| + 2^{|Q|}$ states (while the standard
semi-determinization of~\cite{CourcoubetisY88} has the upper bound of~$4^{|Q|}$).

\begin{lemma}\label{lem:semidet}
Let~$\aut$ be an elevator automaton with~$n$ states.
Then there exists an SDBA with $\bigO(2^n)$ states with the same language.
\end{lemma}

\section{Proofs of Section~\ref{sec:rank-based-compl}}

\theMaxRank*

\begin{proof}
  This proof is a modification of \cite{KupfermanV01} to BAs with mixed transition/state-based acceptance condition.
  We prove by induction that for every $i \geq 0$ there exists $l_i$ such that for all $l \geq l_i$, $\dagof{\alpha}^{2i}$
  contains at most $|Q| - i$ vertices of the form $(q, l)$.
  The proof for the base case $i=0$ follows from the definition of $\dagof{\alpha}^0$, because all levels $l \geq 0$ have at most $|Q|$ vertices of the form $(q, l)$.
  Assume that the hypothesis holds for $i$. We prove it for $i+1$.
  First, if $\dagof{\alpha}^{2i}$ is finite, then it is clear that $\dagof{\alpha}^{2i+1}$ and $\dagof{\alpha}^{2i+2}$ are empty as well.
  If $\dagof{\alpha}^{2i}$ is not finite, then there must be some endangered vertex in $\dagof{\alpha}^{2i+1}$.
  Assume, by way of contradiction, that $\dagof{\alpha}^{2i}$ is infinite and no vertex in $\dagof{\alpha}^{2i+1}$ is endangered.
  Since $\dagof{\alpha}^{2i}$ is infinite, $\dagof{\alpha}^{2i+1}$ is also infinite and therefore each vertex in $\dagof{\alpha}^{2i+1}$ has at least one successor.
  Consider some vertex $(q_0, l_0)$ in $\dagof{\alpha}^{2i+1}$. Since, by the assumption, it is not endangered, there exists an accepting vertex $(q'_0, l'_0)$ reachable from $(q_0, l_0)$ or an accepting edge $(q''_0, l''_0) \xrightarrow{a} (q'_0, l'_0)$ with both vertices $(q''_0, l''_0)$ and $(q'_0, l'_0)$ reachable from $(q_0, l_0)$.
  Let $(q_1, l_1)$ be a~successor of $(q'_0, l'_0)$. By the assumption, $(q_1, l_1)$ is also not endangered and hence, there exists a~vertex $(q'_1, l'_1)$ that is accepting and reachable from $(q_1, l_1)$ or an accepting edge $(q''_1, l''_1) \xrightarrow{a} (q'_1, l'_1)$ with both $(q''_1, l''_1)$ and $(q'_1, l'_1)$ reachable from $(q_1, l_1)$.
  Let $(q_2, l_2)$ be a~successor of $(q'_1, l'_1)$. We can continue similarly and construct an infinite sequence of vertices $(q_1, l_1), (q'_1, l'_1), (q_2, l_2), \ldots, (q_j, l_j), (q'_j, l'_j), (q_{j+1}, l_{j+1}), \ldots$ such that $(q_{j+1}, l_{j+1})$ is a~successor of $(q'_j, l'_j)$ and $(q'_j, l'_j)$ is an accepting vertex or there is a path from $(q_j, l_j)$ to $(q'_j, l'_j)$ containing an accepting edge.
  Such a~sequence corresponds to a~path which contains an accepting vertex or an accepting edge infinitely often, and which is therefore not accepting, which contradicts the assumption that $\alpha \not\in \langof{\aut}$.
  Let $(q, l)$ be an endangered vertex in $\dagof{\alpha}^{2i+1}$.
  Since $(q, l)$ is in $\dagof{\alpha}^{2i+1}$, it is not finite in $\dagof{\alpha}^{2i}$. There are therefore infinitely many vertices reachable from $(q, l)$ in $\dagof{\alpha}^{2i}$. Hence, by König's Lemma, $\dagof{\alpha}^{2i}$ contains an infinite path $(q, l), (q_1, l+1), (q_2, l+2), \ldots$.
  For all $k \geq 1$, the vertex $(q_k, l+k)$ has infinitely many vertices reachable from it in $\dagof{\alpha}^{2i}$ and thus, it is not finite in $\dagof{\alpha}^{2i}$. Therefore, the path $(q, l), (q_1, l+1), \ldots$ exists also in $\dagof{\alpha}^{2i+1}$.
  Since $(q, l)$ is endangered, all the vertices $(q_k, l+k)$ - reachable from $(q, l)$ - are endangered as well. They are therefore not in $\dagof{\alpha}^{2i+2}$.
  For all $j \geq l$, the number of vertices of the form $(q, j)$ in $\dagof{\alpha}^{2i+2}$ is strictly smaller than in $\dagof{\alpha}^{2i}$, which, by the induction hypothesis, is $|Q|-i$. So there are at most $|Q|-(i+1)$ nodes of the form $(q, j)$ in $\dagof{\alpha}^{2i+2}$.
  It is now clear that $\dagof{\alpha}^{2|Q|}$ is empty and therefore it holds that $rank(\alpha) \leq 2|Q|$.

  \qed
\end{proof}

\medskip
\begin{lemma} \label{lemma:2inFKV06}
  There is a~level $l \geq 0$ such that for each level $l' > l$, and for all odd ranks $j$ up to the maximal odd rank, there is a~vertex $(q, l')$
  such that $rank(q, l') = j$.
\end{lemma}


\begin{proof}
  This proof is a modification of \cite{FriedgutKV06} (Lemma 2) to BAs with mixed transition/state-based acceptance condition.
Let $k$ be the minimal index for which $\mathcal{G}_{2k}$ is finite. For every $0 \leq i \leq k-1$, the DAG $\mathcal{G}_{2i+1}$ contains
an endangered vertex.
Let $l_i$ be the minimal level such that $\mathcal{G}_{2i+1}$ contains an endangered vertex $(q, l_i)$.
Since $(q, l_i)$ is in $\mathcal{G}_{2i+1}$, it is not finite in $\mathcal{G}_{2i}$. Thus, there are infinitely many vertices in $\mathcal{G}_{2i}$
that are reachable from $(q, l_i)$. Hence, by König's Lemma, $\mathcal{G}_{2i}$ contains an infinite path $(q, l_i), (q_1, l_i + 1), (q_2, l_i + 2), \ldots$.
For all $j \geq 1$, the vertex $(q_j, l_i + j)$ has infinitely many vertices reachable from it in $\mathcal{G}_{2i}$ and thus, it is not finite in $\mathcal{G}_{2i}$.
Therefore, the path $(q, l_i), (q_1, l_i + 1), (q_2, l_i + 2), \ldots$ exists also in $\mathcal{G}_{2i+1}$.
Since $(q, l_i)$ is endangered, all vertices $(q_j, l_i + j)$ are endangered as well.
It follows that for every $0 \leq i \leq k-1$ there exists a~level $l_i$ such that for all $l' \geq l_i$, there is an endangered vertex $(q, l')$
from which it cannot be reached to any accepting edge and for which $rank(q, l')$ would therefore be $2i + 1$.
\qed
\end{proof}

\medskip
\begin{lemma} \label{lemma:tight}
  There is a~level $l \geq 0$ such that for each level $l' > l$, the level ranking that corresponds to $l'$ is tight.
\end{lemma}

\begin{proof}
This proof is a modification of \cite{FriedgutKV06} (Lemma 3) to BAs with mixed transition/state-based acceptance condition.
According to Lemma \ref{lemma:2inFKV06}, there is a~level $l_1 \geq 0$ such that for all levels $l_2 \geq l_1$ there is a~vertex $(q, l_2)$ such that $rank(q, l_2) = j$ for all odd ranks $j$ up to the maximal odd rank.
In order to be tight, a~level ranking must have an odd rank.
Since even ranks label finite vertices, only a~finite number of levels $l_2$ have even ranks greater than maximal odd rank. Therefore there exists a~level $l \geq l_2$ such that all $l' \geq l$ have odd rank.
\qed
\end{proof}

\medskip
\theComplCorr*

\begin{proof}


This proof is a modification of Schewe's proof from \cite{Schewe09} to BAs with mixed transition/state-based acceptance condition.

First, we show that $\mathcal{L}(\algschewe(\aut)) \subseteq \Sigma^\omega \setminus \langof{\aut}$.
Consider some $\alpha \in \langof{\algschewe(\aut)}$. Let $\rho = S_0 \ldots S_k (S_{k+1}, O_{k+1}, f_{k+1}, i_{k+1}) \ldots$
be a~run of $\algschewe(\aut)$ on $\alpha$ and $\rho' = q_0 q_1 \ldots$ be a~run of $\aut$ on $\alpha$.
From the construction of $\algschewe(\aut)$ it holds that $f_l \transconsist_S^a f_{l+1}$ for all $l > k$.
The sequence $f_{k+1}(q_{k+1}) \geq f_{k+2}(q_{k+2}) \geq \ldots$ is decreasing, and stabilizes eventually.
There is therefore some $m > k$ and $r < 2|Q|$ s.t. $f_p(q_p) = r$ for all $p \geq m$.
If $r$ is even, then there exists some $l > k$ such that for all $l' \geq l$ it holds that $i_{l'} = r$ and $O_{l'} \neq \emptyset$, which contradicts the assumption
that $\rho$ is an accepting run and must therefore contain some state with $O = \emptyset$ infinitely often.
If $r$ is odd, then from some position in the run there is no accepting state of $\aut$ (because accepting states have even rank), and also no accepting transition, because two states $q$ and $q'$ such that $q' \in \delta_F(q,a)$ cannot have the same odd rank. $\rho$ is therefore non-accepting.

Now we show that $\Sigma^\omega \setminus \langof{\aut} \subseteq
\mathcal{L}(\algschewe(\aut))$. Consider some $\alpha \in \Sigma^\omega
\setminus \langof{\aut}$. If there is no run of $\aut$ on $\alpha$, then there
is an accepting run $\rho = S_0 S_1 \ldots $ of $\algschewe(\aut)$ on $\alpha$,
which accepts $\alpha$ in the waiting part of the automaton. Otherwise, there is
a super-tight~run $\rho' = S_0 \ldots S_k (S_{k+1}, O_{k+1}, f_{k+1}, i_{k+1})
\ldots$ of $\algschewe(\aut)$ on $\alpha$ matching the levels of $\dagw$
(see~\cite{HavlenaL2021} for a definition of super-tight run). The existence of
such a run is established by Lemma~\ref{lemma:tight} and
Lemma~\ref{lemma:maxRank}. Since $\alpha \in \Sigma^\omega \setminus
\langof{\aut}$, there is no run with infinitely many accepting states or
transitions on $\alpha$. The ranking of each run therefore stabilizes on some
odd rank, which means that $\rho'$ contains some accepting state infinitely
often and is therefore accepting. If this were not true, then there would be
some $l > k$ such that for all $l' \geq l$ it holds $O_{l'} \neq \emptyset$ and
$i_{l'} = i_l$. This contradicts the fact that the ranking of each run
stabilizes on some odd rank. $\rho'$ is therefore accepting.
\qed


\end{proof}

\medskip
\theTrubRed*
\begin{proof}
	We prove that $\langof{\algschewe(\aut)} = \langof{\but}$.
  The inclusion $\langof{\algschewe(\aut)} \supseteq \langof{\but}$ is clear (we
  just omit certain states from the automaton), so let us focus on the inclusion
  $\langof{\algschewe(\aut)} \subseteq \langof{\but}$.
  Consider some $\word \in \langof{\algschewe(\aut)}$.
  From \cref{lemma:maxRank} and
  \cref{the:schewe-corr} we have that $\word \notin \langof{\aut}$ and
	there is a~ranking of~$\dagw$ having the rank at most~$2|Q|$.
  Since~$\mu$ is a~\trub wrt~$\word$, there is some~$\ell \in \omega$ such
  that~$\ell$ is tight and $\forall k \geq \ell\colon \mu(\levelw(k)) \geq
  f^\word_k$.
  We can use this fact to construct an accepting run~$\rho$ of~$\but$ on~$\word$
  as follows:
  \begin{equation*}
    \rho = \underset{\text{waiting part}}{\underbrace{\levelw(0)\ltr{\word_0}\levelw(1) \ltr{\word_1} \cdots
    \ltr{\word_{\ell-1}} \levelw(\ell)}} \ltr{\word_{\ell}}
    \underset{\text{tight part}}{\underbrace{(\levelw(\ell+1),
    \emptyset, f^\word_{\ell + 1}, 0) \ltr{\word_{\ell + 1}} \cdots}}
  \end{equation*}
  %
	Hence, $\word \in \langof{\but}$. \qed
\end{proof}

\section{Proofs of Section~\ref{sec:elevator}}

\medskip
\theLemElevCorr*

\begin{proof}
	Consider some elevator automaton $\aut$.
	Let $\dagw$ be a run DAG over some word $\word \notin\langof{\aut}$ and $C$ be a component of
	$\aut$. We say that $C$ is terminal in $\dagw$ if for each $\forall i \in
	\omega \forall q \in C: (q,i) \in \dagw \Rightarrow \reach_{\dagw}(q,i)
	\subseteq \{ (c,j) \mid c \in C, j \in \omega \}$.

	\begin{claim}\label{claim:iwa}
		Let $C$ be a terminal $\mathsf{IWA}$ component in $\dagw$. Then, all
		vertices labelled by $C$ will have the rank 0.
	\end{claim}
	\begin{itemize}
		\item Since $\word \notin\langof{\aut}$, all vertices labelled by a state
		from $C$ are finite in $\dagw$ (otherwise the word is accepted).
	\end{itemize}

	\begin{claim}\label{claim:d}
		Let $C$ be a terminal $\mathsf{D}$ component in $\dagw$. Then, all vertices
		labelled by $C$ will have the rank at most 2.
	\end{claim}
	\begin{itemize}
		\item We prove that $\dagw^2$ contains only finite vertices labelled by $C$.
		\item If it is not true, either $\word \in\langof{\aut}$ or $C$ is not
		terminal deterministic (both are contradictions).
	\end{itemize}


	Now we prove the main lemma. First, observe that after application of any rule
	we have that $\mathsf{D, IWA}$ components have an even rank and $\mathsf{N}$
	components have an odd rank. The lemma we prove by induction on given rules.
	In particular, we prove that if a state $q$ was assigned by rank $k$,
	$\dagw^{k+1}$ does not contain node lebelled by $q$.
	\begin{itemize}
		\item Base case: If a terminal component $C$ is $\mathsf{IWA}$, from
		Claim~\ref{claim:iwa} we obtain all states from $C$ will have the rank 0. If
		a terminal component is $\mathsf{D}$, then from Claim~\ref{claim:d} we have
		that all states from $C$ will have the rank bounded by 2.

		\item Inductive case: Assume that for all states $q$ from
      already processed components, if $q$ was assigned by rank $m$,
  	$\dagw^{m+1}$ does not contain node lebelled by $q$.
      Note that inside each rule we can investigate cases
		$\mathsf{D, N, IWA}$ separately since the adjacent components do not affect
		each other.
		\begin{enumerate}[(a)]
			\item[\cref{fig:elev-iwa}] Observe that after $\ell = \max\{ \ell_D, \ell_N+1, \ell_W \}$
			steps of the ranking procedure, in the worst case, all vertices labelled
			by $C$ in $\dagw^\ell$ are finite (otherwise it is a contradiction with
			induction hypothesis). Therefore, in $\dagw^{\ell+1}$ there are no vertices
			labelled by $C$. The ranks of vertices labelled by $C$ is hence $\max\{
			\ell_D, \ell_N+1, \ell_W \}$.

			\item[\cref{fig:elev-det}] We prove that in $\dagw^{\ell}$ all vertices labelled by $C$ are
			finite ($\ell$ is from the rule). From the induction hypothesis, after $\ell - 1$ steps (in the
			worst case) all vertices labelled by adjacent $\mathsf{D, IWA}$ components
			are finite in $\dagw^{\ell}$ (provided that the transitions are
			deterministic). Vertices labelled by adjacent $\mathsf{N}$ components are
			not present in $\dagw^{\ell}$. Therefore, if there is some vertex $v$
			labelled by $C$ in $\dagw^{\ell}$ which is not finite, the only
			possibility is that for each $v' \in \reach_{\dagw^{\ell}}(v)$ we reach in
			$\dagw^{\ell}$ from $v'$ some vertex labelled by a state from the
			$\mathsf{D, IWA}$ components. However, it is a contradiction with the
			transition determinism.

			\item[\cref{fig:elev-nondet}] We prove that in $\dagw^{\ell}$ all vertices labelled by $C$ are
			endangered. This follows from the fact that after $\ell - 1$ steps no
			vertex labelled by the adjacent $\mathsf{D, IWA}$ components is present in
			$\dagw^{\ell}$ (induction hypothesis). Therefore, all vertices labelled by
			$C$ in $\dagw^{\ell}$ are endangered.
		\end{enumerate}
	\end{itemize}
	\qed
\end{proof}

\medskip
\theLemElevRank*

\begin{proof}
  In the worst case an elevator $\aut$ consists of a chain of deterministic
  components connected with nondeterministic transitions. Therefore, using the
  rule from \cref{fig:elev-det}, the maximum rank is increased by 2 for every
  component, which gives the upper bound.
\qed
\end{proof}

\theLemElevBound*

\begin{proof}
	Let $\ell$ be the number of possible odd ranks, i.e., $\ell = \lceil
	\frac{r}{2} \rceil$ where $r$ is the rank upper bound.
	Consider a $\sofi$ for a fixed $i$. We use the same reasoning as
	in~\cite{Schewe09} to compute the number of macrostates $\sofi$ for a fixed
	$i$. A macrostate $\sofi$ we can encode by a function $g: Q \to \{-2, -1, 0,
	\dots, r\}$ s.t. $g(q) = -2$ if $q \notin S$, $g(q) = -1$ if $q \in O$, $g(q)
	= f(q)$ otherwise. Since we consider only tight rankings, $g$ is either onto
	$\{ -2, -1, 1, \dots, j \}$, $\{ -1, 1, \dots, j \}$, $\{ -2, 1, \dots, j \}$,
	$\{ 1, \dots, j \}$ where $j$ is a maximum rank of $\sofi$. For sufficiently big $n$ the number of macrostates for a fixed $i$ is hence
	bounded by
	$$
		4\cdot\stirl{n}{\lceil \frac{j}{2} \rceil + 2}.
	$$
	Since, the maximum ranking can range from $1$ to $r$ we have the bound on the complement size as
	$$
		|\aut| \leq 2^n + 4\ell\cdot\sum_{j=1}^\ell\stirl{n}{j+2}.
	$$
	Moreover, for sufficiently large $n$ (and fixed $\ell$) we have
	$$
		\sum_{j=1}^\ell\stirl{n}{j+2} \leq \ell\cdot\stirl{n}{\ell+2} \leq \ell\cdot\frac{(\ell + 2)^n}{(\ell + 2)!}.
	$$
	The previous follows from the fact that
	$$
		\stirl{n}{\ell} \leq \stirl{n}{\ell + 1}\quad\wedge\quad \stirl{n}{\ell} \sim \frac{\ell^n}{\ell!}
	$$
	for sufficiently large $n$ and fixed $\ell$.
	Therefore,
	\begin{equation}
		\nonumber
		|\aut| \leq 2^n + 4\ell\cdot\sum_{i=1}^\ell\stirl{n}{i+2} \leq 2^n + 4\ell^2\cdot\frac{(\ell + 2)^n}{(\ell + 2)!} \leq 2^n + \frac{(r+m)^n}{(r+m)!},
	\end{equation}
	where $m = \max\{0, 3 - \lceil \frac{r}{2} \rceil\}$, for sufficiently big
	$n$, since $\ell + 2 < r + m$ for all $r \geq 2$.\qed
\end{proof}

\medskip
\theLemDeElLang*
\begin{proof}
	First assume that $\M$ is a set of MSCCs of $\aut$. We prove the first
	inclusion, the second one is done analogically. Consider some $\alpha \in
	\langof{\aut}$. Then, there is a run $\rho = q_0q_1\cdots$ on $\alpha$.
	Moreover, there is some $\ell\in\omega$ and $C \in\M$ s.t. $\rho_k \in C$ for
	all $k \geq \ell$. We can hence construct a run $\rho' = q_0\cdots q_{\ell -
	1}\overline{q_{\ell}}\overline{q_{\ell+1}}\cdots$ over $\alpha$ in
	$\algdeelev(\aut)$. Since $\rho$ is accepting, so $\rho'$ is.
	\qed
\end{proof}

\medskip
\theLemElevCompl*
\begin{proof}
	First, assume that $\aut$ is transformed into $\algdeelev(\aut)$ having $2n$
	states with the rank bounded by 3. From Lemma~\ref{lem:rank-bound} we have
	$$
	 	|\aut| \leq 2^{2n} + \frac{(r+1)^{2n}}{(r+1)!}
	$$
	for $r \geq 3$ and sufficiently big $n$. Therefore,
	$$
		|\aut| \leq 2^{2n} + \frac{4^{2n}}{5!} \in \bigO(16^n).
	$$
	\qed
\end{proof}

\medskip
\theLemGenElevCorr*
\begin{proof}
	Consider some BA $\aut$. Let $\dagw$ be a run DAG over some word $\word
	\notin\langof{\aut}$. We use, in this proof, claims and notation introduced in
	the proof of Lemma~\ref{lem:elevator-corr}.

	\begin{claim}\label{claim:g}
		Let $C$ be a terminal $\mathsf{G}$ component in $\dagw^{2k + 1}$. Then, all
		vertices labelled by $C$ will have the rank at most $2k + 2|C\setminus \accstates|$.
	\end{claim}
	\begin{itemize}
		\item Since $2k + 1 > 0$, $\dagw^{2k + 1}$ does not contain any finite vertices.

		\item Since $C$ is a terminal component, there is some $i \in \omega$ s.t.
		$\forall j > i: |\level_{\dagw^{2k + 1}}(j)\cap C| < |\level_{\dagw^{2k + 2}}(j)\cap C|$ (if we remove an
		endangered vertex, we decrease the width of the run DAG from some level at
		least by 1).

		\item Moreover, since endangered vertices do not contain accepting states,
		the previous observation can be refined to $|(\level_{\dagw^{2k + 1}}(j)\cap
		C)\setminus \accstates| < |(\level_{\dagw^{2k + 2}}(j)\cap C)\setminus \accstates|$.

		\item If we apply the reasoning multiple times, we get that in
		$\dagw^{2k + 2|C\setminus \accstates|}$ remains only finite vertices labelled by a state
		from $C$, therefore the rank is at most $2k + 2|C\setminus \accstates|$.
	\end{itemize}
	%


	Now we prove the main lemma. First, observe that after application of any rule
	we have that $\mathsf{D, IWA,G}$ components have an even rank and $\mathsf{N}$
	components have an odd rank. The lemma we prove by induction on certain
	rules:.In particular, we prove that if a state $q$ was assigned by rank $k$,
	$\dagw^{k+1}$ does not contain node lebelled by $q$.
	\begin{itemize}
		\item Base case: If a terminal component $C$ is $\mathsf{IWA}$, from
		Claim~\ref{claim:iwa} we obtain all states from $C$ will have the rank 0. If
		a terminal component is $\mathsf{D}$, then from Claim~\ref{claim:d} we have
		that all states from $C$ will have the rank bounded by 2. If a terminal
		component is $\mathsf{G}$, from Claim~\ref{claim:g} we have that all states
		from $C$ will have the rank bounded by $2|C\setminus \accstates|$.

		\item Inductive case:
    Assume that for all states $q$ from already processed components, if $q$ was
    assigned by rank $m$, $\dagw^{m+1}$ does not contain node lebelled by $q$.
		\begin{enumerate}[(a)]
		\item[\cref{fig:gen-elev-iwa}] Observe that after $\ell = \max\{ \ell_D, \ell_N+1, \ell_W, \ell_G \} - 1$
			steps of the ranking procedure, in the worst case, all vertices labelled
			by $C$ in $\dagw^\ell$ are finite (otherwise it is a contradiction with the
			induction hypothesis). Therefore, in $\dagw^{\ell+1}$ there are no vertices
			labelled by $C$. The ranks of vertices labelled by $C$ is hence $\max\{
			\ell_D, \ell_N+1, \ell_W, \ell_G \}$.

			\item[\cref{fig:gen-elev-det}] We prove that in $\dagw^{\ell}$ all vertices labelled by $C$ are
			finite ($\ell$ is from the rule). From the induction hypothesis, after $\ell - 1$ steps (in the
			worst case) all vertices labelled by adjacent $\mathsf{D, IWA}$ components
			are finite in $\dagw^{\ell}$ (provided that the transitions are
			deterministic). Vertices labelled by adjacent $\mathsf{N}$ components are
			not present in $\dagw^{\ell}$. Vertices labelled by adjacent $\mathsf{G}$
			components are finite in $\dagw^{\ell}$. Therefore, if there is some
			vertex $v$ labelled by $C$ in $\dagw^{\ell}$ which is not finite, the only
			possibility is that for each $v' \in \reach_{\dagw^{\ell}}(v)$ we reach in
			$\dagw^{\ell}$ from $v'$ some vertex labelled by a state from the
			$\mathsf{D, IWA}$ components. However, it is a contradiction with the
			transition determinism.

			\item[\cref{fig:gen-elev-nondet}] We prove that in $\dagw^{\ell}$ all vertices labelled by $C$ are
			endangered. This follows from the fact that after $\ell - 1$ steps no
			vertex labelled by the adjacent $\mathsf{D, IWA, G}$ components is present in
			$\dagw^{\ell}$ (induction hypothesis). Therefore, all vertices labelled by
			$C$ in $\dagw^{\ell}$ are endangered.

			\item[\cref{fig:gen-elev-gen}] From the induction hypothesis we have that in $\dagw^{\ell}$ where $\ell =
			\max\{ \ell_D, \ell_N+1, \ell_W, \ell_G \}$ all vertices labelled by
			adjacent $\mathsf{D, IW, G}$ components are finite. Vertices labelled by
			adjacent $\mathsf{N}$ components are not present in $\dagw^{\ell}$.
			Therefore, $C$ is terminal in $\dagw^{\ell+1}$. From Claim~\ref{claim:g} we
			have that the rank of $C$ is bounded by $\ell + 2|C\setminus \accstates|$.
		\end{enumerate}
	\end{itemize}
	\qed
\end{proof}

\section{Proofs of Section~\ref{sec:rank-propagation}}

\theOutTrub*

\begin{proof}
	Let $\word \notin \langof \aut$ and $\dagw$ be the run DAG of~$\aut$
  over~$\word$.
  Further, let us use $\mu' = \variant{\mu}{\{S \mapsto \upd_{\mathit{out}}(\mu,
  S, R_1, \ldots, R_m)\}}$.
	\begin{enumerate}
    \item  There are finitely many~$i \in \omega$ such that $\levelwof i = S$.
      Let~$k$ be the last level of~$\dagw$ where~$S$ occurs (or~$0$ if~$S$ does
      not occur on any level of~$\dagw$). Then we can set the~$\ell$ in the
      definition of a~\trub in~\eqref{eq:trub} to be the least $\ell > k$ such
      that~$\ell$ is a~tight level (by Lemma~\ref{lemma:tight} there exists such
      an~$\ell$). Then the condition holds trivially.
    \item  There are infinitely many~$i \in \omega$ such that $\levelwof i = S$.
      Then, since~$\mu$ is a~\trub, let~$\ell$ be the~$\ell$ in~\eqref{eq:trub}
      for which~$\mu$ satisfies~\eqref{eq:trub}.
      We need to show that for every~$k > \ell$ such that $\levelwof k = S$, it
      holds that $\mu'(S) \geq f^\word_k$.
      Let~$\cP \subseteq \{R_1, \ldots, R_m\}$ be the set of predecessors
      of all occurrences of~$S$ on~$\dagw$ below~$\ell$, i.e., for all $k >
      \ell$ such that $\levelwof k = S$, we have $\levelwof{k-1} \in \cP$.
      Since the ranks of levels in a~run DAG are lower for levels that are
      higher, it is sufficient to consider only the first such a~$k$.
      Let~$R$ be the predecessor of~$S$ at~$k$, i.e., $R = \levelwof {k-1}$.
      Since we do not know which particular $R_j \in \cP$ it is, we need to
      consider all~$R_j \in \cP$.
			Since $k-1$ is already a tight position, we have that $\mu'' =
			\variant{\mu}{\{S \mapsto \max\{ \mu(R_1), \ldots, \mu(R_m)  \} \}}$ is
			a~\trub for the same $\ell'' = \ell$ in~\eqref{eq:trub} (recall that all
			tight positions have the same odd rank). Further, $\mu$ is also a~\trub,
			therefore, $\mu' = \variant{\mu}{S\mapsto \min\{ \mu''(S), \mu(S) \}}$ for
			$\ell$.\qed
  \end{enumerate}
\end{proof}

\medskip
\theOutTerm*

\begin{proof}
	Let $\mu$ be a~\trub and $\mu' = \variant \mu \{S \mapsto
	\upd_{\mathit{out}}(\mu, S, R_1, \ldots, R_m)\}$. From
	Lemma~\ref{lem:out-trub} we have that $\mu'$ is a \trub as well, which means
	that starting from $\mu_0$ using $\upd_{\mathit{out}}$ we get \trub{s} only.
	Moreover, $\mu(P) \geq \mu'(P)$ and $\mu'(P) \geq 0$ for each $P \in 2^Q$.
	The fixpoint evaluation hence eventually stabilizes.
\qed
\end{proof}

\medskip
\theInnerTrub*

\begin{proof}
  Let $\word \notin \langof \aut$ and $\dagw$ be the run DAG of~$\aut$
  over~$\word$.
  Further, let us use $\mu' = \variant{\mu}{\{S \mapsto \upd_{\mathit{in}}(\mu,
  S, R_1, \ldots, R_m)\}}$.
	First, we prove that the following claim:

	\smallskip
	\begin{claim}\label{claim:trub-meet}
		Let $\mu_1, \mu_2$ be two \trub{}s wrt $\alpha$. Then $\mu'$,
    defined as $\mu'(S) := \mu_1(S) \sqcap \mu_2(S)$ is a \trub wrt $\alpha$.
	\end{claim}
	\begin{claimproof}
		Follows from the definition (with choosing $\ell' = \max\{ \ell_1,
			\ell_2 \}$) where $\ell_1$ is from the definition of a \trub for $\mu_1$, $\ell_2$ is for $\mu_2$.
	\end{claimproof}

	\smallskip
	\noindent
  We need to consider the following two cases:
  \begin{enumerate}
    \item  There are finitely many~$i \in \omega$ such that $\levelwof i = S$.
      Let~$k$ be the last level of~$\dagw$ where~$S$ occurs (or~$0$ if~$S$ does
      not occur on any level of~$\dagw$). Then we can set the~$\ell$ in the
      definition of a~\trub in~\eqref{eq:trub} to be the least $\ell > k$ such
      that~$\ell$ is a~tight level (by Lemma~\ref{lemma:tight} there exists such
      an~$\ell$). Then the condition holds trivially.
    \item  There are infinitely many~$i \in \omega$ such that $\levelwof i = S$.
      Then, since~$\mu$ is a~\trub, let~$\ell$ be the~$\ell$ in~\eqref{eq:trub}
      for which~$\mu$ satisfies~\eqref{eq:trub}.
      We need to show that for every~$k > \ell$ such that $\levelwof k = S$, it
      holds that $\mu'(S) \geq f^\word_k$.
      Let~$\cP \subseteq \{R_1, \ldots, R_m\}$ be the set of predecessors
      of all occurrences of~$S$ on~$\dagw$ below~$\ell$, i.e., for all $k >
      \ell$ such that $\levelwof k = S$, we have $\levelwof{k-1} \in \cP$.
      Since the ranks of levels in a~run DAG are lower for levels that are
      higher, it is sufficient to consider only the first such a~$k$.
      Let~$R$ be the predecessor of~$S$ at~$k$, i.e., $R = \levelwof {k-1}$.
      Since we do not know which particular $R_j \in \cP$ it is, we need to
      consider all~$R_j \in \cP$.
      Let $M = \{\maxrsaof{R_j}{\mu(R_j)} \mid R_j \in \cP, a \in \Sigma\}$.
      Then, since~$\mu$ is a~\trub, $\bigsqcup M$ will also be a~\trub
      \ol{elaborate why?}.
      Moreover, from Claim~\ref{claim:trub-meet},
      $\theta = \mu(S) \sqcap \bigsqcup M$ will also be a~\trub, and so $\theta
      \geq f^\word_k$.
      Then, if the rank of~$\theta$ is even, we can decrease it to the nearest
      odd rank, since tight rankings are, by definition, of an odd rank.
      \qed
  \end{enumerate}
\end{proof}

\medskip
\theInTerm*

\begin{proof}
	Let $\mu$ be a~\trub and $\mu' = \variant \mu \{S \mapsto
	\upd_{\mathit{in}}(\mu, S, R_1, \ldots, R_m)\}$. From
	Lemma~\ref{lem:inner-trub} we have that $\mu'$ is a \trub as well, which means
	that starting from $\mu_0$ using $\upd_{\mathit{in}}$ we get \trub{s} only.
	Moreover, $\mu(P) \geq \mu'(P)$ and $\mu'(P) \geq \{ q \mapsto 0 \mid q \in Q \}$ for each $P \in 2^Q$.
	The fixpoint evaluation hence eventually stabilizes.
\qed
\end{proof}

\fi

\end{document}